\documentclass[11pt]{article}
\makeatletter
\renewcommand\section{\@startsection {section}{1}{\z@}%
{-3.5ex \@plus -1ex \@minus -.2ex}%
{2.3ex \@plus.2ex}%
{\normalfont\fontfamily{phv}\fontsize{16}{19}\bfseries}}
\renewcommand\subsection{\@startsection{subsection}{2}{\z@}%
{-3.25ex\@plus -1ex \@minus -.2ex}%
{1.5ex \@plus .2ex}%
{\normalfont\fontfamily{phv}\fontsize{14}{17}\bfseries}}
\renewcommand\subsubsection{\@startsection{subsubsection}{3}{\z@}%
{-3.25ex\@plus -1ex \@minus -.2ex}%
{1.5ex \@plus .2ex}%
{\normalfont\normalsize\fontfamily{phv}\fontsize{14}{17}\selectfont}}
\makeatother

\usepackage{amsmath,amsthm,amssymb}
\usepackage{graphicx,enumerate,xcolor}
\usepackage{natbib,url,tikz,xr,ulem,authblk}
\usepackage{algorithm,algpseudocode}

\newtheorem{theorem}{Theorem}

\newtheorem{lemma}{Lemma}

\newtheorem*{remark}{Remark}

\usepackage{setspace}      
\begin{document}
\def\spacingset#1{\renewcommand{\baselinestretch}%
{#1}\small\normalsize} \spacingset{1}
%%%%%%%%%%%%%%%%%%%%%%%%%%%%%%%%%%%%%%%%%%%%%%%%%%%%%%%%%%%%%%%%%%%%%%%%%%%%%%
\newcommand{\blind}{0}
\if0\blind{
\title{\bf Transform-Resampled Double Bootstrap Percentile with Applications in System Reliability Assessment}
\author[1,2]{Junpeng Gong}
\author[2]{Xu He} 
\author[2]{Zhaohui Li\thanks{lizh@amss.ac.cn}}
\affil[1]{School of Mathematical Sciences, University of Chinese Academy of Sciences, Beijing, China}
\affil[2]{State Key Laboratory of Mathematical Sciences (SKLMS), Academy of Mathematics and Systems Science, Chinese Academy of Sciences, Beijing, China}
\date{}
\maketitle}\fi
\if1\blind{
\title{\bf  Transform-Resampled Double Bootstrap Percentile with Applications in System Reliability Assessment}
\author{}
\date{}
\vspace{6cm}  
\maketitle
\vspace{3cm}} \fi

\begin{abstract}
System reliability assessment(SRA) is a challenging task due to the limited experimental data and the complex nature of the system structures.  
Despite a long history dating back to \cite{buehler1957confidence}, exact methods have only been applied to SRA for simple systems.
High-order asymptotic methods, such as the Cornish-Fisher expansion, have become popular for balancing computational efficiency with improved accuracy when data are limited, but frequently encounter the "bend-back" problem in high-reliability scenarios and require complex analytical computations.
To overcome these limitations, we propose a novel method for SRA by modifying the double bootstrap framework, termed the double bootstrap percentile with transformed resamples. 
In particular, we design a nested resampling process for log-location-scale lifetime models, eliminating the computational burden caused by the iterative resampling process involved in the conventional double bootstrap.
We prove that the proposed method maintains the high-order convergence property, thus providing a highly accurate yet computationally efficient confidence limit for system reliability. 
Moreover, the proposed procedure is straightforward to implement, involving only a simple resampling operation and efficient moment estimation steps. 
Numerical studies further demonstrate that our approach outperforms the state-of-the-art SRA methods and, at the same time, is much less susceptible to the bend-back issue.
\end{abstract}

\noindent%
{\it Keywords:} System reliability assessment; 
Double bootstrap; Log-location-scale family; 
Bend back.

%\newpage
\spacingset{1.9} % DON'T change the spacing!

\section{Introduction}
In this study, we focus on the problem of system reliability assessment (SRA) using component-level data, a field pioneered by \cite{buehler1957confidence}. 
The reliability of a product is typically defined as the probability that the product will perform satisfactorily up to a specified mission time \citep{henley1981}.
A major task in reliability assessment is determining the lower confidence limit (LCL) that encompasses the true reliability at a specified confidence level \citep{henley1981} using experimental data.
System reliability assessment is a fundamental concept in engineering, physics, and finance, providing a durability index that enhances risk management and competitive advantage \citep{meeker2022statistical}. 

Many modern engineering systems are extremely complex, making system-level lifetime experiments expensive or even intractable. 
In such cases, we must calculate the LCL using component-level data. 
In this paper, we focus on scenarios where the system comprises several components whose interconnections are known from mechanical domain knowledge and engineering features \citep{meeker2022statistical, hong2014confidence,li2020higher}. 
The initial work of \cite{buehler1957confidence}  in the SRA area proposed an exact confidence limit for system reliability.
However, its applications are limited to binomial component models and simple system structures, such as series or parallel systems with a few components. 
The method suffers from a heavy computational burden as the system structure gets complex \citep{du2020exact}.
Therefore, asymptotic approximations have been widely developed for assessing system reliability \citep{meeker2022statistical}.
One of the most popular asymptotic methods for SRA is the delta method \citep{hong2014confidence}, which leverages the asymptotic normality of the maximum likelihood estimator (MLE) of system reliability. 
The delta method provides an efficient and asymptotically accurate LCL for system reliability and performs well in scenarios involving a large number of products under testing. 
% The process begins with the collection of component lifetime data through component-level testing, which informs the estimation of each component's reliability. These estimates are then aggregated, taking into account the structure of the system, to derive an overall system reliability estimate. 

However, in fields such as aerospace engineering and solar cell systems \citep{meeker2022statistical}, the prohibitive cost and time requirements of lifetime testing render large sample sizes unattainable.
In these sectors, small sample sizes are common, which limits the effectiveness of the delta method. 
To improve the accuracy of interval estimation for system reliability in small sample scenarios, \cite{winterbottom1980asymptotic} introduced the WCF method. 
This approach constructs a polynomial transformation of the point estimation for model parameters. 
The coefficients are adjusted to ensure that the distribution of the polynomial converges to the standard normal distribution at a higher order, rather than the estimation itself, as in the delta method. 
This yields a higher-order confidence bound construction using the normal approximation.
The WCF expansion is widely used in constructing the LCL for the reliability of solid-state power controllers \citep{cai2017wcf} and for the population proportion using group testing with misclassification \citep{xiong2017confidence}.
Recent developments have further refined the WCF method.
For example, \cite{yu2007statistical} expanded its application to more complex systems, while \cite{li2020higher} recently adapted it for multi-parameter lifetime models.
% Along the other line, Bayesian methods have historically attracted considerable attention. 
% For example, \cite{martz1990} demonstrated a Bayesian procedure for assessing the reliability of series or parallel systems composed of binomial subsystems. 
% Further exploration by \cite{hamada2004bayesian} elaborated on the utility of Bayesian prediction and tolerance intervals. 
% Furthermore, \cite{li2016design} introduced a method based on hierarchical systems for demonstration testing, highlighting the utility of these methods for small sample sizes in both theoretical and practical applications.
% However, selecting an appropriate prior is crucial for the Bayesian method \citep{guo2013bayesian}, and sampling from the posterior distribution can be a challenging task. 
% Moreover, we discover that even in the simplest scenario, the Bayesian approach produces a systematic bias in terms of the frequentist coverage probability.
% The bias tends to be analytically unavailable in general cases, making it impossible to de-bias.
% These challenges become even more pronounced for complex systems, thereby limiting the practical application of Bayesian methods in the reliability assessment of complex systems.

Despite these developments, challenges remain in SRA.
First, the LCLs produced by the aforementioned asymptotic methods often fall outside the valid interval $[0,1]$ (referred to as the 'falling outside' issue) and exhibit the bend-back issue \citep{Hong2007}. 
The term 'bend-back' refers to a paradoxical phenomenon that the LCL decreases as the true reliability increases, contradicting intuitive expectations.
This issue, first discovered by \cite{Hong2007}, is prevalent in methods used to compute reliability confidence limits.
% This issue can be better understood by examining the nature of reliability itself: the reliability function $R(t)$ is a probability function confined to the interval $[0,1]$ and monotonically decreases as the mission time $t$ increases. Therefore, a reasonable LCL for $R(t)$ should ideally lie within $[0,1]$ and exhibit similar monotonic behavior with increasing $t$. 
% However, both the delta method and the WCF method fail to ensure these properties, especially when $R(t)$ approaches the boundary of $[0,1]$. 
% Moreover, the LCLs produced by these methods often fall outside the range of quantity of interest and lack reasonable monotonicity. 
% \cite{Hong2007} further noted that this paradox arises because the asymptotic variance of estimated reliability blows up in extrapolation.
% The authors then introduced transformation-based methods to alleviate this issue in component-level reliability assessments. However, an effective solution addressing these issues at the general system-level reliability assessment remains elusive.
Second, the aforementioned higher-order asymptotic methods require extensive mathematical derivations for calculating SRA. 
These derivations must be independently developed for each specific system structure. 
The lack of a universal derivation methodology applicable across diverse system structures considerably limits the flexibility and efficiency of these methods in practical engineering applications.

% Firstly, both the WCF and delta methods occasionally yield confidence limits beyond the range of $[0,1]$, an outcome that lacks practical relevance. As discussed in \cite{Hong2007}, this is because the estimates of the variance of system reliability sometimes blow-up. Furthermore, the adjustments that the WCF method makes to the normal approximation do not accurately reflect the true distribution, which can further lead to confidence limits falling outside of the range $[0,1]$. This is particularly problematic in sectors requiring extremely high reliability, such as aerospace and other critical industries, where values such as $0.9999$ are commonplace. Secondly, the complexity of the analytical calculations required by the WCF method could restrict its applicability in real-world scenarios, where system structures and component lifetime models are diverse and complex. 
% Tddress the challenges of reliability assessment in highly reliable systems.

% \cite{Hong2007} also noted that employing an inverse transformation of a proper distribution function could prevent confidence interval endpoints from exceeding the parameter space. However, this transformation might introduce a phenomenon known as "bend-back behavior," where the lower confidence limit for larger true reliability paradoxically becomes smaller. Given the complexity of system structures, providing a universal solution for identifying transformations suitable for system reliability estimations proves challenging.

The bootstrap method \citep{efron1994introduction} offers an alternative approach that can avoid the falling outside and bend-back issues, and seems promising in addressing these challenges in SRA problems. 
The bootstrap method achieves an asymptotically accurate estimation of the LCLs through a resampling process \citep{efron1994introduction}.
%, which is more computationally efficient than sampling from posterior distributions in Bayesian methods. 
Numerous variations of bootstrap methods have been proposed, such as basic bootstrap, also called empirical bootstrap, and bootstrap percentile method, which can be applied within parametric and non-parametric frameworks \citep{martin}. 
Bootstrap methods are typically first-order accurate, i.e., the same as the delta method.
The double bootstrap method \citep{hall1986bootstrap, beran1987prepivoting} introduces an iterated layer of resampling into each initial bootstrap sample, achieving higher-order convergence.
% It can be further extended by progressively adding layers of resampling to reduce coverage errors, which is known as the iterated bootstrap method.
% In this study, we focus on the double bootstrap method as we have observed that additional layers of resampling offer only slight improvements.

Despite the improvements in the convergence order, the bootstrap-type methods have not received much attention in the reliability area. 
The most recent study by \cite{marks2014applying} explored the application of non-parametric bootstrapping in reliability assessments of simple series or parallel two-component systems. 
It was noted that these bootstrap methods can introduce a systematic bias, the extent of which depends on the structural configuration of the system. 
To date, there has been little effort to further explore the use of the advanced version, such as the double bootstrap method, in reliability assessment. 
This lack of progress can be attributed to several challenges.
% Firstly, complex scenarios with heterogeneous data sizes among components and the presence of censored data, which are common in survival analysis, hinder the use of this method.
Firstly, the bootstrap method has not yet been explored for use in scenarios involving heterogeneous data sizes and censored data, both of which are common in reliability analysis.
Secondly, all double bootstrap methods are computationally intensive due to the recursion of the resampling procedure and a large number of point estimation steps.
% For a system with $s$ components, each with failure data of size $n$. 
To illustrate, let $B$ and $C$ denote the number of resamples in the two layers, respectively. 
The double bootstrap method, therefore, requires $BC$ repetitions for sampling the data and point estimation processes. 
According to \cite{booth1994monte,booth1998allocation}, it is recommended to set the value of $BC$ at around $10^6$. 
This implies that the double bootstrap method requires millions of resamples for each component in SRA, which becomes unaffordable for complex systems with lots of components.
Attempts to reduce the computational burden for the double bootstrap methods have been investigated. 
For example, \cite{chang2015} has shown that setting $C=1$ in the second layer can achieve the same order of accuracy as using $C=10^3$ for the bias correction tasks, 
However, none of them are suitable for constructing confidence intervals.

% Direct application of the iterated bootstrap method for system reliability assessment thus presents substantial difficulties.
% Meanwhile, resampling methods like bootstrapping have been explored for system reliability assessment. However, \cite{marks2014applying} have pointed out that bootstrap methods can introduce a systematic bias, the magnitude of which depends on the system's structure. These insights led to our development of the iterated bootstrap method, initially discussed by \cite{hall1986bootstrap,beran1987prepivoting}, for system reliability assessment. This method employs multi-layer resampling to provide more accurate confidence intervals than one layer resampling for a universal problem. It potentially combines the improvement of transformation with the simplicity of bootstrap methods in system reliability assessment.
 
In this article, we propose a computationally efficient double bootstrap framework for SRA that addresses these challenges. 
In particular, we accelerate the resampling process by directly generating the moment-based reliability estimates, bypassing conventional data resampling and estimation steps.
Furthermore, we transform and reuse these estimates across the second bootstrap layer. 
We call the proposed method \textit{Double Bootstrap Percentile with Transformed resamples} (DBPT).
The primary contributions of our methods are summarized as follows. 
First, the DBPT is computationally efficient compared to the conventional double bootstrap method.
The computation cost is reduced from $O(snBC)$ to $O(s(B+C)+sBC)$ where $s$ denotes the number of components and $n$ is the sample size.
Second, our method always offers reasonable LCLs that fall inside the valid range $[0,1]$ and rarely exhibits the bend-back problem, which will be illustrated in the numerical study.
Third,  by using the properties of the moment-based reliability estimates, we achieve the same order acceleration in constructing confidence intervals compared to the acceleration in \cite{chang2015} for bias correction.
Moreover, we prove that the new method maintains high-order asymptotic accuracy, the same as the conventional double bootstrap methods.
% Second, we have optimized the resampling process by directly drawing the moment-based reliability estimates instead of resampling the data and computing the corresponding estimates. 
% Third, we share these estimates across the second layer resamples. As a result, the number of required resampling and computation iterations has been reduced from $O(BCn)$ to $O(B+C)n$, and in best case reducing to $O(B+C)$.
% By using the properties of the moment-based reliability estimator, we achieve the same order acceleration in constructing confidence intervals compared to the acceleration in \cite{chang2015} for bias correction.
Finally, to handle censored data, we successfully apply our method after completing the data using an imputation algorithm \citep{xiao2014study}.
In summary, the proposed method addresses the challenges of applying the double bootstrap method in SRA while avoiding the bend-back and falling outside issues, which are commonly encountered in other higher-order methods.
This approach consistently provides a highly accurate LCL for complex system reliability assessment for various sample sizes of component failure data, as demonstrated in Section \ref{sec:numerical_studies}.

The structure of this article is organized as follows. 
Section~\ref{sec:background} introduces the background of system reliability assessment. 
Section~\ref{sec:bootstrap_method} details our approach and the theoretical results. 
Section~\ref{sec:numerical_studies} displays numerical studies demonstrating the effectiveness of our method. 
Section~\ref{sec:conclusions} provides a summary of our contributions, as well as a discussion of limitations and future research directions.

\section{System Reliability Assessment}
\label{sec:background}
This section provides the necessary preliminaries, including the definition of the system, the component lifetime model, and reviews the conventional methods for SRA.
% \subsection{Notation}
% The notation used throughout our discussion is defined as follows:
% \begin{align*}
% t & : \text{mission time} \\
% z_{1-\alpha} & : 1 - \alpha \text{ quantile of the standard normal distribution} \\
% r_i(t) & : \text{reliability of the } i\text{-th component at time } t \text{ (an unknown constant)} \\
% \hat{r}_i(t) & : \text{estimated reliability of the } i\text{-th component at time } t \\
% n_i & : \text{number of observed failure times for } i\text{-th component} \\
% R(t) & : \text{system reliability at time } t \text{ (an unknown constant)} \\
% \hat{R}(t) & : \text{estimated system reliability at time } t \\
% \mathbf{r}(t) & : \text{the vector whose $i$-th entry is $r_i(t)$.}
% \end{align*}
% \lzhcmt{the structure of data should be emphasized, put it at the beginning of some subsection.}

\subsection{System models and component failure-time distributions}

In this article, we focus on the reliability assessment of a coherent system. 
A coherent system comprising $s$ components operates based on the following three fundamental rules. 
First, at any given mission time $t$, each component and the system are either functional or non-functional, e.g., failure. 
Second, each component has a non-negligible impact on the system, i.e., enhancing the reliability of any individual component contributes to improving the system's reliability. 
Third, the components are operated independently. 
For a coherent system, the system reliability at a specific mission time $t$, $R(t)$, can be expressed as 
\begin{equation*}
    R(t) = \text{P}(T >t) = \psi(r_1(t),\ldots,r_s(t)),
\end{equation*}
where $T$ denotes the failure time of the system, $\text{P}(T >t)$ gives the probability that the system remains functional by time $t$, $r_i(t)$ denotes the reliability of the $i$-th component at $t$, and the system structure function $\psi$ can be determined by incorporating mechanical domain knowledge and engineering features. 
% This function maps from the $ s $-dimensional unit hypercube $[0, 1]^s$ to $[0, 1]$. For coherent systems where all $ x_i $'s are within (0,1), $ \psi $ is strictly increasing and twice differentiable, as discussed in \cite{barlow1981statistical}.
For instance, some system structures can be represented by the system diagram shown in Figure \ref{fig:system}.
The structure functions are 1) $\psi (x_1, \ldots, x_s) = \prod_{i=1}^s x_i$ for the series system with $s$ components, 2) $\psi(x_1, \ldots, x_s) = 1 - \prod_{i=1}^s (1 - x_i)$  for the parallel system with $s$ components, and 3) $\psi(x_1, x_2, x_3, x_4) = \{1 - (1 - x_1)(1- x_2)\}\{1 - (1 - x_3)(1-x_4)\}$ for the $2 \times 2$ series-parallel system with component-level redundancy.

% Since $\psi$ is generally a known function, although possibly complex, the task of estimating the reliability of the system can be reduced to estimating the reliability of its individual components.
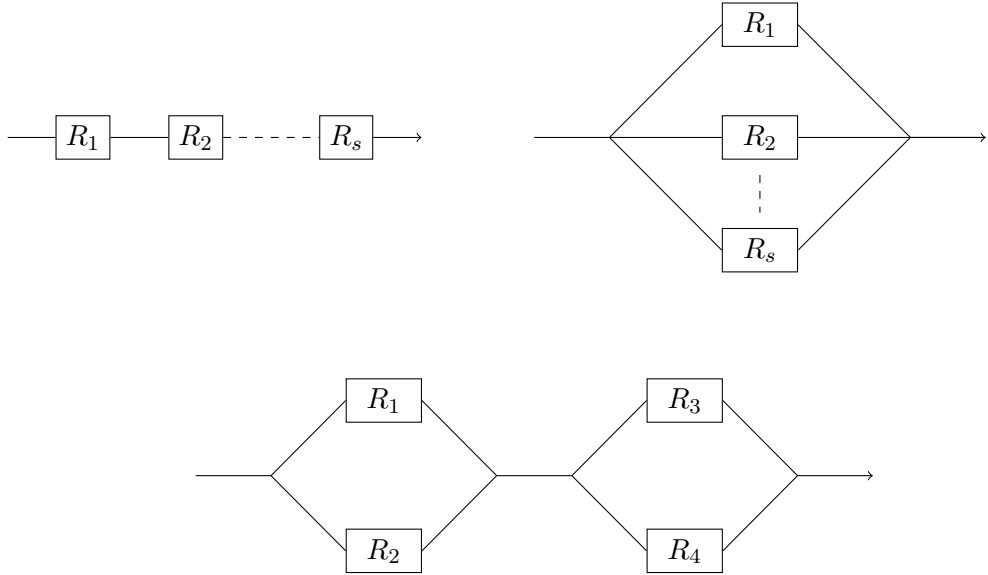
\begin{figure}
\centering
\begin{tikzpicture}
% series
    % Drawing the boxes and their labels
    \node[draw, rectangle] (a) at (-7,0) {$R_1$};
    \node[draw, rectangle] (b) at (-5.5,0) {$R_2$};
    \node[draw, rectangle] (c) at (-3.5,0) {$R_s$};
    
    % Drawing the line
    \draw[-] (-8,0) -- (a);
    \draw[-] (a) -- (b);
    \draw[dashed] (b) -- (c);
    \draw[->] (c) -- (-2.5,0);

% parallel
    % Drawing the input and output lines
        \draw (-1,0) -- (0,0);
        \draw (4,0) -- (5,0);
    
        % Drawing the boxes for components
        \node[draw, rectangle, minimum width=1cm, minimum height=0.5cm] (a) at (2,1.5) {$R_1$};
        \node[draw, rectangle, minimum width=1cm, minimum height=0.5cm] (b) at (2,0) {$R_2$};
        \node[draw, rectangle, minimum width=1cm, minimum height=0.5cm] (c) at (2,-1.5) {$R_s$};
    
        % Drawing lines from the input to each component
        \draw (0,0) -- (a.west);
        \draw (0,0) -- (b.west);
        \draw (0,0) -- (c.west);
    
        % Drawing lines from each component to the output
        \draw[->]  (4,0) -- (5,0);
        \draw (a.east) -- (4,0);
        \draw (b.east) -- (4,0);
        \draw (c.east) -- (4,0);
        % \draw (a.east) -- (4,1.5) -- (4,0) -- (5,0);
        % \draw (b.east) -- (4,0);
        % \draw (c.east) -- (4,-1.5) -- (4,0);
    
        % Optional: Dashed line to indicate continuation of components
        \draw[dashed] (2, -0.5) -- (2, -1);

% 2\times 2 series-parallel system
        \node [draw, rectangle, minimum width=1cm, minimum height=0.5cm] (a) at (-3,-3.5) {$R_1$};
        \node [draw, rectangle, minimum width=1cm, minimum height=0.5cm] (b) at (-3,-5.5) {$R_2$};
        
        \node [draw, rectangle, minimum width=1cm, minimum height=0.5cm] (c) at (1,-3.5) {$R_3$};
        \node [draw, rectangle, minimum width=1cm, minimum height=0.5cm] (d) at (1,-5.5) {$R_4$};
        
        \draw (-5.5, -4.5) -- (-4.5, -4.5); % Input line to subsystem A
        \draw (-1.5, -4.5) -- (-0.5, -4.5); % Connecting line between subsystems
        \draw[->] (2.5, -4.5) -- (3.5, -4.5); % Output line from subsystem B
        % Sub A
        \draw (-4.5, -4.5) -- (a.west);  
        \draw (a.east) -- (-1.5,-4.5);
        \draw (-4.5, -4.5) -- (b.west);
        \draw (b.east) -- (-1.5,-4.5);

        % Sub B
        \draw (-0.5, -4.5) -- (c.west);  
        \draw (c.east) -- (2.5,-4.5);
        \draw (-0.5, -4.5) -- (d.west);
        \draw (d.east) -- (2.5,-4.5);
\end{tikzpicture}
\caption{System diagrams of a $s$-series (top left), $s$-parallel (top-right) and $2\times 2$ series-parallel (bottom) system.}
\label{fig:system}
\end{figure}

% The failure time of individual components is typically modeled by a random variable whose distribution is selected from the log-location-scale family, which includes many commonly used lifetime distributions, eg., the Weibull, log-normal, and log-logistic distributions \citep{meeker2022statistical}. 
% That is, the distribution function for failure times $T_i$ of the $i$-th component can be written as $F_i((\log t - \mu_i)/\sigma_i)$, where 
% $\mu_i$ and $\sigma_i$ are the unknown model parameters.
% The reliability function $r_i(t)$ is then given by
% \begin{equation}
% \label{equ:reliability}
%     r_i(t) = P(T_i \geq t) = 1 - F_i \left(\frac{\log t - \mu_i}{\sigma_i}\right).
% \end{equation}

% \begin{example}[Log-normal distribution]
%     When $T_i$ follows a log-normal distribution, $F_i(x)$ is the standard normal cumulative distribution function; 
%     when $T_i$ follows a Weibull distribution, $F_i(x) = 1 - e^{-e^x}$.
% % Remark that the distribution of each $T_i$ can vary; i.e., $F_i(x)$ corresponding to $X_i$ may differ among components. 
% \end{example}

For each component, its lifetime is typically modeled as a random variable $T_i$.
The distribution of $T_i$ is assumed to follow some parametric distributions, summarized from engineering knowledge. 
Let $G_{\theta_i}(x)$ and $g_{\theta_i}(x)$ denote the cumulative distribution function (cdf) and probability density function (pdf), respectively, for the lifetime of the $i$-th component. 
The reliability of the $i$-th component is 
\begin{equation*}
    r_i(t) = P(T_i \geq t) = 1 - G_{\theta_i}(t).
\end{equation*}
Let $\boldsymbol{\theta} = (\theta_1, \ldots, \theta_s)$ represent the unknown parameters of the system.
The system reliability $R(t)$, a function of time $t$ parameterized by $ \boldsymbol{\theta} $, is sometimes denoted as $R_{\boldsymbol{\theta}}(t)$ to emphasize the dependence on $\boldsymbol{\theta}$. 

Given the lifetime data collected at the component lever, i.e., the data is $\mathcal{D} = \cup_{i=1}^s \mathcal{X}_i = \{t_{i,j}: i=1,...,s, j=1,...,n_i\}$, where $t_{i,j}$ gives the $j$th failure time of the $i$th component.
Throughout this paper, we assume that the data from different components are mutually independent. 
The goal of SRA is to establish a lower confidence limit $ R_L(t) $ for $ R(t) $ with a confidence level of $1-\alpha$, i.e., finding the largest possible $R_L(t)$ ensuring that 
\begin{equation}
    \text{P}(R(t)\geq R_L(t)) \geq 1-\alpha.
\end{equation}

%Clearly, $R(t)$ is bounded within $[0,1]$ and is a monotonically decreasing function of $ t $. 
%It makes sense that an LCL $ R_L(t) $ should lie within $[0,1]$ and decrease monotonically as $ t $ increases.
% The construction of a confidence interval for $R(t)$ proceeds by first finding a function of both the sample and the parameter, a root $R_n(X_n; \theta_F)$, whose distribution is either known or can be consistently estimated. If $\hat{\theta}_n(X_n)$ is an estimator of $\theta_F$ based on the sample $X_n$, familiar choices for the root are 
% \begin{equation*}
% R_n(X_n; \theta_F) = \sqrt{n} (\hat{\theta}_n(X_n) - \theta_F)
% \end{equation*} 
% A confidence interval for $\theta_F$ can then be derived by estimating the appropriate quantiles by inverting $J_n(\cdot; \theta_F)$, the cumulative distribution function of $R_n(X_n; \theta_F)$.

\subsection{The normal approximation procedures}
\label{subsec:delta}

%Despite the early development of exact confidence interval methods \cite{buehler1957confidence}, approximation methods with favorable convergence properties and greater computational efficiency have become increasingly popular.
%The theoretical foundation for most approximation methods is the asymptotic property of point estimation of system reliability. 
We herein introduce one of the most important approximation SRA methods, namely the delta method \citep{hong2014confidence}. 
Given data $\mathcal{D}$, the likelihood for $\boldsymbol{\theta}$ is given by 
\begin{equation}\label{eq:likelihood}
    L( \boldsymbol{\theta}\mid \mathcal{D} ) = \prod_{i=1}^{s} \prod_{j=1}^{n_i}  g_{\theta_i}(t_{ij}).
\end{equation}
Let $\boldsymbol{\hat \theta}$ denote the maximum likelihood estimator (MLE) of $\boldsymbol{\theta}$ obtained by maximizing \eqref{eq:likelihood}. The MLE of $R(t)$ is then $\hat R(t) = R_{\boldsymbol{\hat \theta}}(t)$.  For simplicity, we assume $n_1=\cdots=n_s=n$ hereinafter. 
Due to the Delta method, it is well known that $S_n = \sqrt{n}(\hat{R}(t) - R(t))$ converges in distribution to $N(0,\tau^2)$.
The asymptotic variance is $\tau^2 = \left( \nabla_{\boldsymbol{\theta}} R  \right)^T I^{-1}_{\boldsymbol \theta} \left( \nabla_{\boldsymbol{\theta}} R \right)$, where $I_{\boldsymbol \theta} =  \mathbb{E}\left\{ \nabla_{\boldsymbol{\theta}} \log L( \boldsymbol{\theta}\mid \mathcal{D} ) \cdot (\nabla_{\boldsymbol{\theta}} \log L( \boldsymbol{\theta}\mid \mathcal{D} ))^T   \right\}$ is the Fisher's information.
Here, $\nabla_{\boldsymbol{\theta}} R$ denotes the gradient of $R(t)$ with respect to (w.r.t.) $\boldsymbol{\theta}$ and $I_{\boldsymbol{\theta}}$ denotes the fisher information matrix.
In practice, a plug-in estimator for $\tau$ is given by replacing $\boldsymbol{\theta}$ by its estimation $\boldsymbol{\hat{\theta}}$, i.e., $\hat{\tau}^2 = \left( \nabla_{\boldsymbol{\theta}} R  \mid_{\boldsymbol{\theta}= \boldsymbol{\hat{\theta}}}  \right)^T I^{-1}_{\boldsymbol{\hat{\theta}}} \left( \nabla_{\boldsymbol{\theta}} R\mid_{\boldsymbol{\theta}= \boldsymbol{\hat{\theta}}} \right)$.
Then the $1-\alpha $ asymptotic upper confidence bound for $\hat{R}(t)$ is immediately obtained by 
\begin{equation}
    R_{\text{D}}(t) = 2\hat{R}(t) - z_{1-\alpha}\hat{\tau},
\end{equation}\label{eq:delta_var}
where $z_a$ is the $a$-upper quantile of the standard normal distribution.
It is derived from Edgeworth expansion that $R_{\text{D}}(t)$ is first-order asymptotically accurate, i.e., $\text{P}(R(t) \geq R_{\text{D}}(t)) = 1 - \alpha + O(n^{-1/2})$.
Note that the confidence set is said to be asymptotically accurate in the k-th order if its coverage probability remainder term is $O(n^{-k/2})$ \citep{shao2008mathematical}.

Despite its wide application, the delta method suffers from three major limitations in SRA.
First, the coverage probability errors of order only $O(n^{-1/2})$ can be insufficient in small sample scenarios. 
Second, the delta method uses a symmetric, unbounded normal distribution approximation for the distribution of $S_n$ while the actual distribution of $S_n$ is bounded and typically asymmetric. 
Therefore, $R_D$ often falls outside of the valid interval $[0,1]$ because the normal approximation fails to respect the bounded nature of $S_n$. 
Although \cite{Hong2007} introduced transformation-based methods to alleviate this issue in component-level reliability assessments,  an effective solution addressing these issues at the general system-level reliability assessment remains elusive.
Third, $R_D$ often exhibits the bend-back phenomenon, where the LCL paradoxically fails to decrease monotonically with $t$, as shown in Fig.~\ref{fig:bend back}. 
This behavior arises from instability in the standard error estimator $\hat{\tau}$ \citep{Hong2007}.

To address the limited accuracy of the delta method, higher-order asymptotic methods based on the Cornish-Fisher expansion are proposed \cite{li2020higher,winterbottom1980asymptotic}, achieving refined LCLs with reduced coverage errors of order at least $O(n^{-1})$.
These methods are implemented by constructing a polynomial transformation that adjusts $S_n$ toward normality via the C-F based expansion.
However, they require cumbersome analytical calculations in the explicit derivation of high-order cumulants to construct the polynomial.
Currently, no automated derivation procedures for general system structures exist to streamline this process, posing significant implementation challenges in real applications.
Moreover, these high-order methods are still struggling with the bend-back and falling outside issues.
In the remainder of this paper, we will introduce a novel method that addresses these three challenges simultaneously.

\begin{figure}
    \centering
    \includegraphics[width = 0.6\textwidth]{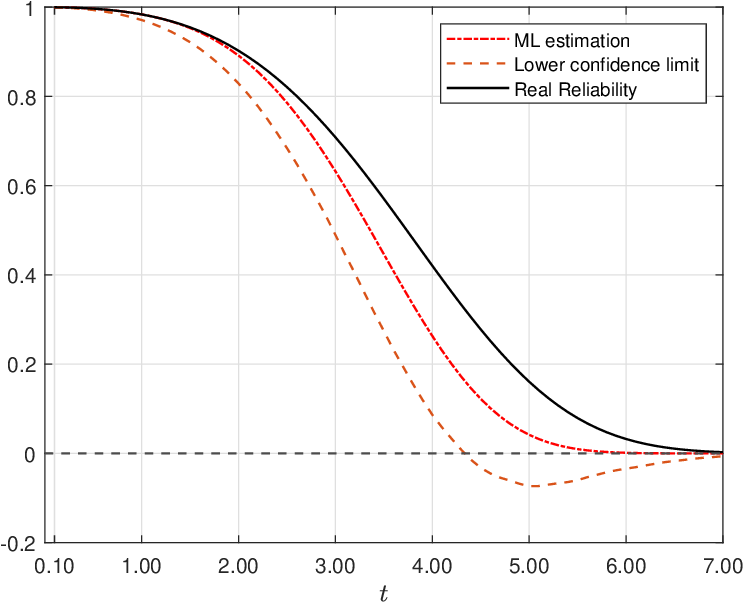}
    \caption{The real reliability (black solid line), ML estimation (red dash-dot line), LCL by delta method (red dashed line) 
    % and LCL by our proposed method 
    for a series system composed of three Weibull-type components with a sample size of $n=10$.}
    \label{fig:bend back}
\end{figure}

\section{Efficient Double Bootstrap Method for System Reliability Assessment}
\label{sec:bootstrap_method}
In this section, we propose a novel double bootstrap method for SRA. 
To ensure the logical flow of our presentation, Section \ref{subsec:boot_intro} first introduces how to apply bootstrap methods to SRA. 
Building on this, Section \ref{subsec:double_boot} discusses the potential and limitations of the conventional double bootstrap method. 
To address the challenges posed by the double bootstrap method, Section \ref{subsec:moment-estimation} introduces a moment-based estimation for the log-location-scale family, forming the basis for our approach. 
Finally, Section \ref{subsec:dbp} summarizes the proposed double bootstrap percentile method with transformed resamples for SRA.

\subsection{Basic bootstrap and bootstrap percentile method}\label{subsec:boot_intro}
The bootstrap methods \citep{efron1994introduction} have been extensively used in various inference tasks \citep{davison1997bootstrap,CensorSchneider}.
% There are two major streams, namely the parametric bootstrap and the nonparametric bootstrap. 
% \sout{Our focus is on parametric bootstrap methods, as they perform better with limited sample sizes\lzhcmt{citation}.}
% In our framework, to effectively manage cases where data follow the log-location-scale family, we recommend employing parametric bootstrap methods.
Two common parametric bootstrap methods used to construct confidence intervals are the basic bootstrap \citep{diciccio1996bootstrap} and the bootstrap percentile method \citep{martin}.
In this section, we briefly introduce these approaches for the SRA.

%For simplicity, assume the sample sizes $ n_i $ are equal across all components, i.e., $ n_i = n $ for $ i = 1, \ldots, s $. 
Denote the point estimation of system reliability as $\hat{R}(t) = \psi(\hat{r}_1,\ldots,\hat{r}_s)$, 
where $\hat{r}_i(t) = 1 - G_{\hat{\theta}_i}(t)$ is some point estimate for $r_i(t)$, $\hat{\theta}_i$ are some estimations for $\theta_i$ from the original data $\mathcal{X}_i$.
We first introduce how to obtain the bootstrap sample, which is the key step in implementing both bootstrap methods. 
First, generate $B$ samples $\mathcal{X}_{i,1}^*,\ldots,\mathcal{X}_{i,B}^*$ from the distribution $G_{\hat{\theta}_i}(t)$, which serves as an approximation to $G_{\theta_i}(t)$ from which the data $\mathcal{X}_i$ is drawn. 
Each of the samples, $\mathcal{X}_{i,j}^*$, has the same sample size as $\mathcal{X}_i$.
Next, we compute the bootstrap estimates $ \hat{\theta}_{i,j}^* $ and thus $\hat{r}_{i,j}^*(t)$ from the samples $ \mathcal{X}_{i,j}^* $ for $ i = 1, \ldots, s $ and $ j = 1, \ldots, B $.  
These component-level bootstrap estimates $\hat{r}_{i,j}^*$ are then aggregated to produce bootstrap estimates of the system reliability $ \hat{R}_j^* = \psi(\hat{r}_{1,j}^*,\ldots,\hat{r}_{s,j}^*)$.
The collection $ \hat{R}_1^*, \ldots, \hat{R}_B^* $ is called the \textit{bootstrap sample}, which forms an empirical approximation to the sampling distribution of $ \hat{R}(t) $.
%Both bootstrap approaches produce the LCLs based on the bootstrap sample, based on different statistical principles.
% Let $ \lceil x \rceil $ denote the smallest integer greater than or equal to $ x $, and order the  $ \hat{R}_{(j)} $'s as $ \hat{R}_{(1)}^* \leq \hat{R}_{(2)}^* \leq \ldots \leq \hat{R}_{(B)}^* $.

Given the bootstrap sample, the $100(1-\alpha)\%$ \textit{bootstrap percentile LCL} for $ R(t) $ is immediately obtained by the sample quantile of the same level from the bootstrap sample, i.e., 
$ R_{\text{BP}} = \hat{R}_{(\lceil B\alpha \rceil)}^*$, where $ \lceil x \rceil $ denotes the smallest integer greater than or equal to $ x $ and $\hat{R}_{(a)}^*$ denote the $a$-th smallest value among $\hat{R}_1^*,\ldots,\hat{R}_B^*$.
This method rests on the principle that the empirical distribution of bootstrap estimates approximates the distribution of $\hat{R}(t)$. 
Consequently, the  $\hat{R}_{\text{BP}}$ is expected to cover $R(t)$ with probability approximately $1-\alpha$, mirroring the probability with which $R_{BP}$ covers $\hat{R}(t)$. 

On the other hand, the \textit{basic bootstrap method} is given by is $R_{\text{BB}} = 2\hat{R}(t) - \hat{R}^*_{(\lceil B(1-\alpha)\rceil)}$, which calculates the LCL in a form analogous to \ref{eq:delta_var} in delta method.
To illustrate it, consider the test statistic $ S_n = \sqrt{n} (\hat{R}(t) - R(t)) $. 
One can approximate the $1 - \alpha$ quantile of $S_n$'s distribution by finding the sample quantile of its bootstrap sample $ S_{n,j}^* = \sqrt{n} (\hat{R}_{j}^*(t) - \hat{R}(t)), j=1,\ldots, B$, denoted as $S_{n, (\lceil B(1 - \alpha) \rceil)}^*$. 
Rearranging the inequality $S_n \leq S_{n, (\lceil B(1 - \alpha) \rceil)}^*$ yields $R_{BB}$, ensuring $R(t) \geq R_{BB}$ holds with approximately $1 - \alpha$ probability.

Theorem~\ref{thm:BP and BB} establishes the theoretical properties for both the bootstrap percentile and basic bootstrap lower confidence limits.

\begin{theorem}
\label{thm:BP and BB}
Assume $n_i = n$ and the distribution $G_{\theta_i}(t)$ is absolutely continuous and  exists at least fourth-order moments for $i=1,\ldots,s$. 
The bootstrap percentile LCL, $R_{BP}$, is within the range $[0,1]$ and is monotonically decreasing w.r.t. time $ t $. 
As $B, n \rightarrow \infty$, the coverage probability of $R_{BP}$ and $R_{BB}$ are
\begin{align*}
\text{P} \left( R(t) \geq R_{BP} \right) & = 1- \alpha + O(n^{-1/2}),\\
\text{P} \left( R(t) \geq R_{BB} \right) & = 1- \alpha + O(n^{-1/2}).
\end{align*}
\end{theorem}
The proof of Theorem \ref{thm:BP and BB} is provided in the supplementary materials.

Both methods are first-order asymptotically accurate.
However, the LCL obtained via the basic bootstrap method is more likely to encounter issues, including bend-back behavior and invalid confidence limits falling outside $[0, 1]$, for reasons similar to those affecting the delta method. 
In contrast, the bootstrap percentile method avoids these issues since the percentile $\hat{R}_{(\lceil \alpha \rceil)}(t)$ always lies within $[0, 1]$ and maintains monotonicity w.r.t. $t$. 
As a result, we recommend the \textit{bootstrap percentile method} for practical use in SRA.

% The dominant term of the coverage error arises from using the estimator computed by the resamples from the approximate distribution $F_1((\log t - \hat{\mu}_1)/\hat{\sigma}_1)$ instead of the true distribution $F_1((\log t - \mu_1)/\sigma_1)$. 
% For both bootstrap methods, the asymptotic error of the LCL could be eliminated by applying double bootstrap technique\citep{hall1986bootstrap,beran1987prepivoting}. \lzhcmt{the purpose of doing double bootstrap should be posted here}

\subsection{Conventional Double bootstrap percentile method}\label{subsec:double_boot}

As discussed in Section~\ref{sec:background}, first-order methods are typically insufficiently accurate in small sample scenarios for SRA, which is common in reliability engineering for highly reliable products. 
Consequently, %a higher-accuracy method, namely the double bootstrap, is introduced next. 
%While the double bootstrap technique can be built upon both bootstrap methods, we specifically adopt the double bootstrap percentile method to address the SRA problem, as it not only enhances accuracy but also retains the desirable properties of the bootstrap percentile method, avoiding bend-back issue and invalid confidence limits outside $[0,1]$.
%These advantages are demonstrated via numerical simulation in Section \ref{sec:numerical_studies}.
we extend and apply the double bootstrap percentile (DBP) method \citep{martin, hall1986bootstrap, beran1987prepivoting} to the SRA problem and reach a second-order accurate LCL using a nested resampling process. 
In this section, we start our method by introducing the conventional double bootstrap for SRA.
The DBP with nested resampling procedure is proposed in Section \ref{subsec:moment-estimation} and \ref{subsec:dbp}.

The DBP begins by generating the first-level bootstrap sample $ \hat{R}_1^*, \ldots, \hat{R}_B^* $, as in the standard bootstrap percentile method. 
For $i = 1, \ldots, s $ and $j = 1, \ldots, B $, we resample $ \mathcal{X}_{i,j,1}^{**}, \ldots, \mathcal{X}_{i,j,C}^{**} $ from the distribution $ G_{\hat{\theta}_{i,j}^*}(t) $. 
Following the same procedure as obtaining $\hat{R}_{j}^*$, we obtain the second layer bootstrap estimate  $\hat{R}_{j,l}^{**}$ from $\mathcal{D}_{j,l}^{**} = \cup_{i=1}^s \mathcal{X}_{i,j,l}^{**}$ for $j=1,\ldots B$ and $l=1,\ldots,C$.
Let $ \hat{\alpha} = \hat{u}_{( \lceil C\alpha \rceil )}$, where $ \hat{u}_j = C^{-1} \sum_{l= 1}^C I(\hat{R}_{j,l}^{**} \leq\hat{R}), j=1,\dots,B$, $I(\cdot)$ is the indicator function.
%Since $\tilde{\alpha}$ is the $\alpha$ quantile of the distribution of $U(R)$ and the $ \hat{u}_j $'s are the empirical approximation of the distribution of $U(R)$, this suggests taking the $\hat{\alpha}$ as an estimate for the $\tilde{\alpha}$.
The double bootstrap percentile LCL is given by
\begin{equation*}
R_{DBP} = \hat{R}^*_{(\lceil B \hat{\alpha} \rceil)}.
\end{equation*}

The following theorem summarizes the properties of the DBP method.

% theorem 
\begin{theorem}
% [\cite{hall2013bootstrap}]
\label{thm:order of DBP}
Assume $n_i = n$ and the distribution $G_{\theta_i}(t)$ is absolutely continuous and  exists at least fourth-order moments for $i=1,\ldots,s$. 
% As $n \rightarrow \infty$, $B \rightarrow \infty$ and $C \rightarrow \infty$,
The double bootstrap percentile lower confidence limit, $R_{DBP}(t)$, is within $[0,1]$. As $B, C, n \rightarrow \infty$, the coverage probability is
\begin{equation}
\label{equ: order DBP}
% \lim_{B,C \rightarrow \infty}
\text{P} \left( R(t) \geq R_{DBP} \right) = 1- \alpha + O(n^{-1}). 
% C^{-1}(\alpha + 1/2 ) + o(B^{-1/2}  + C^{-1}).
\end{equation}
\end{theorem}

We now explain in detail how the double bootstrap percentile corrects the second-order bias in the bootstrap percentile procedure.
It is shown that the double bootstrap percentile method adjusts the percentile level $\alpha$ of the bootstrap percentile method to $\hat{\alpha}$. 
Denote the distribution function of $\hat{R}_{j}^*$ as $U(x) = \text{P}(\hat{R}_{j}^* \leq x \mid \mathcal{D})$.
% Then we have $U (\hat{R}^*_{( \lceil B\theta \rceil)} ) = \theta$ for any $\theta \in (0,1)$. 
The coverage probability of the  $R_{BP} = \hat{R}^*_{( \lceil B\alpha \rceil)}$ is 
\begin{equation*}
    C(\alpha) = \text{P} \left(R \geq \hat{R}^*_{( \lceil B\alpha \rceil)} \right) = \text{P} \left(U(R) \geq \alpha) \right).
\end{equation*}
Ideally, we favor a percentile form LCL $\hat{R}^*_{( \lceil B\tilde{\alpha} \rceil)}$, where $\tilde{\alpha}$ satisfies $C(\tilde{\alpha}) = 1-\alpha$. To approximate $\tilde{\alpha}$, the DBP method employs a second-layer bootstrap resampling. 
This generates $\hat{u}_j$'s whose empirical distribution approximates the distribution of $U(R)$. 
The $\alpha$ quantile of $\hat{u}_j$'s, i.e., $\hat{\alpha}$, then serves as an estimate of $\tilde{\alpha}$. 
Consequently, the coverage level of $R_{DBP} = \hat{R}^*_{( \lceil B\hat{\alpha} \rceil)}$, $C(\hat{\alpha})$, achieves closer alignment with the target $C(\tilde{\alpha})$ compared to the $C(\alpha)$. 

% This adjustment calibrate the error introduced by the first layer bootstrap. 
% The bootstrap estimates 
% $ \hat{r}_{i,j}^* $ is computed based on resamples $\mathcal{X}_{i,j}^*$ from the estimated distribution $G_{\hat{\theta}}(t)$, while the  $ \hat{r}_i $ is computed based on $\mathcal{X}$ from the true distribution $G_{\theta}(t)$. 
% This discrepancy introduces approximation errors in using the bootstrap estimates 
% $ \hat{R}_{j}^* $ approximate the distribution of  $ \hat{R} $.
% % This approximation is based on the data $ \mathcal{X}_{i,j}^* $ sampled from a wrong distribution $ F_i((\log t - \hat{\mu}_i)/\hat{\sigma}_i) $, rather than the true distribution $ F_i ((\log t - \mu_i)/\sigma_i) $.
% In the second round of resampling, we treat $ F_i((\log t - \hat{\mu}_i)/\hat{\sigma}_i) $ as the true distribution and treat $ \mathcal{D}_j^* $ as the original data, to carry out the bootstrap procedure again. The error introduced by sampling from $ F_i( (\log t - \hat{\mu}_{i,j}^*)/ \hat{\sigma}_{i,j}^* ) $, instead of $ F_i((\log t - \hat{\mu}_i)/\hat{\sigma}_i) $ reflects the original error. 
% This helps reduce the coverage error of the $R_{BP}$.

Despite the advantages of the double bootstrap percentile method from an accuracy perspective, the significant computational demand caused by the nested resampling process prevents its practical application in SRA.
Specifically, it requires generating $B(C+1)$ resamples of $\mathcal{X}_i$, each requiring a corresponding estimate $\hat{r}_i$ to be computed. 
Therefore, the total time consumption is $O(snBC)$ for SRA with $s$ components. 
Typically, the value of $BC$ should be at least $10^5$ for accepted performance \citep{booth1994monte}, resulting in a minimum of $10^5$ resampling and computation. 
This becomes particularly challenging in SRA when MLE is adopted for estimating component reliability. 
For lifetime models, such as the Weibull distribution, a closed-form for MLE often lacks and is obtained through optimization algorithms. 
Consequently, applying the double bootstrap percentile method in SRA involves repeating  $10^5$ times of resampling and optimization process.
For systems with large $s$ (e.g., complex systems with many components), this process becomes computationally prohibitive, even with modern computing resources.

To address these challenges, we herein modify the double bootstrap percentile method from two perspectives:
First, we propose a moment-based estimation method for estimating the parameters in the log-location-scale family, introduced in Section \ref{subsec:moment-estimation}, which is efficient to compute.
Second, we modify the second layer of the double bootstrap percentile method by replacing the large number of resampling-estimation steps with the re-estimation steps directly, thus reducing the computational cost.
This modification is introduced in detail in Section \ref{subsec:dbp}.

\subsection{Moment-based estimation for log-location-scale family}\label{subsec:moment-estimation}
% For the parametric bootstrap method, point estimation of model parameters plays a vital role as they are evaluated in each resampling process.

In this section, we introduce an efficient method for estimating component reliability $r_i(t)$ based on moments of the log-location-scale distribution, which generalizes the results in \cite{li2023optimal,li2016design}.

Unlike results given in Sections~\ref{subsec:boot_intro} and~\ref{subsec:double_boot}, in this subsection we specifically assume the component data follow the log-location-scale distribution. 
This family of distributions includes widely used distributions in reliability analysis, such as the Weibull, log-normal, and exponential distributions \citep{meeker2022statistical}.
Let $X_i = \log T_i$ denote the log-transformed random variable of $i$-th component lifetime and $F_i((x-\mu_i)/\sigma_i)$ denote the cdf of $X_i$. 
The reliability of the $i$-th component is 
\begin{equation}
\label{equ:reliability}
    r_i(t) = \text{P} (T_i \geq t) = \text{P} (X_i \geq \log t) = 1 - F_i ((\log t - \mu_i)/\sigma_i).
\end{equation}
The population mean and variance for $X_i$ is $\mathbb{E}(X_i) = \mu_i + \kappa_{i1} \sigma_i, \mathrm{var} (X_i) = \kappa_{i2}^2 \sigma_i^2$, where $\kappa_{i1}$ and $\kappa_{i2}^2$ denote the population mean and variance of the distribution $F_i(x)$.

Consider the log-transformed lifetime data $x_{i,j} = \log t_{i,j}, j = 1,\ldots,n_i$.   
Let $\bar{x}_{i} = n_i^{-1}\sum_{j=1}^{n_i} x_{i,j}$
and $s_i^2 = (n_i-1)^{-1} \sum_{j=1}^{n_i} (x_{i,j} - \bar{x}_i)^2$ denote the sample mean and sample variance, respectively. The moment based estimator for $\sigma_i$ and $\mu_i$ is 
\begin{equation}
\label{equ:mu and sigma}
     \hat {\sigma}_i = s_i/ \kappa_{i2}, \quad \hat{\mu}_i = \bar{x}_i -  \kappa_{i1} \hat {\sigma}_i.
\end{equation}
By substituting the estimators $ \hat{\mu}_i $ and $ \hat{\sigma}_i $ for $ \mu_i $ and $ \sigma_i $ in \eqref{equ:reliability}, the moment based estimator of $r_i(t)$ is given by
\begin{equation}
\label{equ:estimator}
    \hat{r}_i(t) = 1 - F_i ((\log t - \hat{\mu}_i)/\hat{\sigma}_i).
\end{equation}
According to the law of large numbers, the estimators $\hat{\mu}_i$ and $\hat{\sigma}_i$ are consistent, thereby ensuring the consistency of $\hat{r}_i(t)$ as well. 
This yields an analytical formula for obtaining the reliability estimate of components, as opposed to a numerical solution. 
Furthermore, we discover that the moment estimator $\hat{r}_i(t)$ depends on the unknown parameters $\mu_i$ and $\sigma_i$ solely through $r_i(t)$, thus a single-parameter random variable, which is presented in the following theorem.
\begin{theorem}
\label{thm:mme}
Assume $t_{i,j},j=1,\ldots,n_i$ follow the log-location-scale distribution $F_i((\log t - \mu_i)/\sigma_i)$. Let $\hat{r}_i(t)$ denote the moment-based estimator derived from \eqref{equ:mu and sigma} and \eqref{equ:estimator}. 
Then the distribution of $\hat{r}_i(t)$ only depends on the true reliability $r_i(t)$ and $F_i(x)$, whose form is given by
\begin{equation}
    \label{equ:MME}
    \hat{r}_i(t) \sim 1 - F_i \left[ \left\{F_i^{-1}(1 - r_i(t)) -  \bar{Z}_{n_i}\right\}  \kappa_{i2} /M_{n_i} + \kappa_{i1} \right],
\end{equation}
where $ \bar{Z}_{n_i} $ and $ M_{n_i}^2 $ are the sample mean and variance of $ n_i $ log-transformed samples $\{x_i = \log t_i, i=1,2,...,n_i\}$, $ \kappa_{i1} $ and $ \kappa_{i2}^2 $ are the population mean and variance of the distribution $ F_i(x) $, respectively.
%Here, $\kappa_{i1}$ and $\kappa_{i2}$ is the population mean and variance of distribution $F_i(x)$, respectively. 
% For the bootstrap version $\hat{r}_i^*(t)$ based on $t_{i,j}^*, j =1,\ldots,n_i$, sampled from $F_i((\log t - \hat{\mu}_i)/\hat{\sigma}_i)$, we have
% \begin{equation}
%     \label{equ:MME-BP}
%     \hat{r}_i^*(t) \sim_{\text{i.i.d}} 1 - F_i \left[ \left\{F_i^{-1}(1 - \hat{r}_i(t)) -  \bar{Z}_{n_i}\right\} M_{n_i}^{-1} \kappa_{i2} + \kappa_{i1} \right].
% \end{equation}
\end{theorem}
% Then $$\bar{X}_n = \sigma^{-1} \left(n^{-1} \sum_{i=1}^{n} \log T_i -\mu \right), $$ and $$\bar{S}_n = (n-1)^{-1} \sigma^{-2} \sum_{i=1}^{n} \left( \log T_i - \sum_{i=1}^n n^{-1}\log T_i \right)^2.$$
% For example, suppose that $ T_1, \ldots, T_n $ iid follow a  Weibull distribution with scale parameter $ \eta $ and shape parameter $ \beta $, where the cumulative distribution function (cdf) is
% $$
%  F_W(t; \eta, \beta) = 1 - \exp\left(-\left(\frac{t}{\eta}\right)^\beta\right), \quad t > 0, \eta > 0, \beta > 0.
% $$ 
% Then, $ X_i = \log T_i $ follows the extreme-value distribution with location parameter $ \mu  $ and scale parameter $ \sigma $, where the cdf is
% $$
%  F_0(x) = 1 - \exp\left(-\exp\left(\frac{x - \mu}{\sigma}\right)\right).
% $$
% In this scenario, $z_1 = \gamma$ is the Euler–Mascheroni constant and $z_2^2 = \pi^2/6$.
% % and the reliability at $ t_0 $ is $ R_W(t_0) = \exp\left(-\exp\left(\frac{\log t_0 - \mu}{\sigma}\right)\right) $.
The proof is provided in the supplementary materials.
Note that assuming the component lifetime distribution belongs to the log-location-scale family is widely accepted, as a wide range of lifetime distributions are encompassed by this family. 
Several concrete examples will be given below.

% Similarly, the distribution of bootstrap estimate $\hat{r}_i$ depends on the $\hat{r}_i$.
% Consequently, the dependence of the bootstrap version of the moment-based estimator, $\hat{r}^*_i(t)$, on $\hat{r}_i(t)$ mirrors the dependence of $\hat{r}_i(t)$ on $r_i(t)$. 
\begin{remark}
The following three types of lifetime distributions belong to the log-location-scale family.
Thus, we could express the moment estimation for reliability as a single-parameter random variable. \\
\textbf{Weibull} type component, i.e., $X_i = \log T_i$ follows the distribution $F_i((x - \mu_i)/\sigma_i)$ with $F_i(x) = F_e(x)$ is the extreme distribution function. The population mean and variance are $\kappa_{i1} = -\gamma, \kappa_{i2} = \pi/\sqrt{6}$, where $\gamma = 0.5772...$ is the Euler's constant.
The moment-based estimator is
\begin{equation*}
\hat{r}_i(t) \sim  \exp\left( -\exp\left( \frac{\log (-\log r_i(t)) - \bar{Z}_{n_i}}{M_{n_i}} \cdot \frac{\pi}{\sqrt{6}} - \gamma \right) \right)
\end{equation*} \\
\textbf{Log-normal} type component, i.e., $X_i = \log T_i$ follows the normal distribution with c.d.f $\Phi((x-\mu_i)/\sigma_i)$, where $\Phi(x)$ is the c.d.f of the standard normal distribution. We have $\kappa_{i1} = 0, \kappa_{i2} = 1$. 
The moment-based estimator is
\begin{equation*}
% \hat{r}_i(t) \sim \Phi \left\{ T_{n_i} + \sqrt{(n_i-1)/S_{n_i}} \Phi^{-1}(r_i(t)) \right\},
\hat{r}_i(t) \sim 1 - \Phi \left\{ (\Phi^{-1} \left( 1 - r_i(t)\right) - \bar{Z}_{n_i})/M_{n_i}  \right\},
\end{equation*}
where $\bar{Z}_{n_i} / M_{n_i}$ follows a Student's t-distribution $t(n_i-1)$ and $(n_i-1)M^2_{n_i}$ follows a Chi-squared distribution $\chi^2(n_i-1)$.\\
% where $T_{n_i}$ follows the Student's t-distribution $t(n_i-1)$ and $S_{n_i}$ follows a chi-squared distribution $\chi^2(n_i-1)$ and are independent.\\
\textbf{Exponential} type component with distribution function $P(T_i \leq t) = 1 - \text{e}^{-\lambda_i t}$ is a special case where the moment estimate for $\lambda_i$ can be easily obtained from the fact $\mathbb{E}(t_{i1}) = 1/\lambda_i$, which is $\hat{\lambda}_i = n_i / (\sum_{j=1}^{n_i} t_{i,j})$. The corresponding moment-based estimator is
\begin{equation*}
    \hat{r}_i(t) \sim \exp\{\log r_i(t) / M_{n_i}\},
\end{equation*}
where $M_{n_i}$ follows the Gamma distribution $\Gamma(n_i,n_i)$. 
% It corresponds the \eqref{equ:MME} with $\kappa_{i1} =0, \kappa_{i2} =1$ and $Z_{n_i} = 0$.
\end{remark}

\subsection{The double bootstrap percentile with transformed resamples for log-location-scale family}\label{subsec:dbp}

In this section, we introduce our modified double bootstrap percentile method for SRA, termed the double bootstrap percentile with transformed resamples. 
We adopt the common assumption in SRA that component lifetimes follow log-location-scale distributions. 
Within this framework, the estimate $\hat{r}_i(t)$ now specifically denotes the moment-based estimate obtained through Eqs.~\eqref{equ:mu and sigma} and \eqref{equ:estimator}. 

% The notation $\hat{r}_{i,j}^*$ and $\hat{r}_{i,j,l}^{**}$ continue to denote first and second layer bootstrap estimates, respectively.
To improve the computational efficiency, we would not generate the second layer bootstrap sample $\hat{r}_{i,j}^*$ and $\hat{r}_{i,j,l}^{**}$ through resampling of $\mathcal{X}_{i,j}^*$ and $\mathcal{X}_{i,j,l}^{**}$.
We discover that those can be computed directly from $ \bar{Z}_{n_i}^{*j}, M_{n_i}^{*j} $ and $ \bar{Z}_{n_i}^{**l}, M_{n_i}^{**l} $ via \eqref{equ:MME}. 
Algorithm~\ref{al:type-generating} details the generation of distribution-specific statistics $\bar{Z}_{n_i}, M_{n_i}$ for log-normal, exponential and general log-location-scale components. 
The general log-location-scale case requires the standardized cumulative distribution function as an additional input. For instance, the standardized CDF for the Weibull case is $F_e(x)$.

Embedding this modification to the double bootstrap percentile method, we propose the \textit{double bootstrap percentile method with transformed resamples} (DBPT), which is given in Algorithm~\ref{al:dbp}. 
Focusing on a simple case where the system consists of only one component, Fig.~\ref{fig:tik} provides a concrete representation and comparative visualization of DBPT against DBP. 
Our approach achieves computational improvements over DBP methods through two key modifications. 
First, in steps $10$ and $18$ of Algorithm~\ref{al:dbp}, we compute $\hat{r}_{i,j}^*$ and $\hat{r}_{i,j,l}^{**}$ via a shortcut formula based the sampling of $\bar{Z}_{n_i}^{*j}, M_{n_i}^{*j}$ and $\bar{Z}_{n_i}^{**l}, M_{n_i}^{**l}$. 
This reduces the computation cost from $O(snBC)$ in DBP to $O(sBC)$.
Second, in step 14, we generate one set of $\bar{Z}_{n_i}^{**l}, M_{n_i}^{**l}$ for each $i, l$ and reuse to compute $\hat{r}_{i,j,l}^{**}$ across $j = 1,\ldots, B$ in step 18. 
This is actually equivalent with transforming the resamples $\mathcal{X}_{i,1,l}^{**}$ to generate $\mathcal{X}_{i,j,l}^{**}$ for $j=2,\ldots,B$ in conventional DBP method.
The computation cost is $O(sC + sBC)$ compared to DBP's $O(snBC)$. 
Overall, our method achieves a total computational cost of $O(s(B+C)+sBC)$, outperforming the $O(snBC)$ of DBP methods. 
% This makes our approach more practical since the $B, C$ are usually large.

\begin{algorithm}
\caption{Generating $\bar{Z}_{n_i}$ and $M_{n_i}$ for different lifetime distributions}
\label{al:type-generating}
\begin{algorithmic}[1]
\State \textbf{Input: } Sample size $n_i$ and component type (either log-normal, exponential, or general log-location-scale, the latter requiring its standardized CDF $F_i(x)$ as additional input).
\State \textbf{Output: } auxiliary statistics $ \bar{Z}_{n_i} $, $ M_{n_i} $.

\If{component is \textbf{log-normal type}}
    \State Independently generate $ T_{n_i} \sim t(n_i-1) $ and  $ U_{n_i} \sim \chi^2(n_i-1) $; 
    \State $ M_{n_i} \gets \{(n_i-1) /U_{n_i} \}^{1/2}$;
    \State $\bar{Z}_{n_i} \gets T_{n_i} \cdot M_{n_i}$;
\ElsIf{Component is \textbf{Exponential type}}
    \State Generate $M_{n_i}$ from Gamma distribution $\Gamma(n_i, 1/n_i)$ and $\bar{Z}_{n_i} = 0$;
\ElsIf{component is \textbf{general log-location-scale type}}
% \lzhcmt{I think this algorithm narrows the applicability of our method, for exponential and log-normal, you can indeed propose specific method for calculating the moment estimation, but for the Weibull distribution, it should be adapted to the general log-location scale family. Otherwise, it looks like we can only handle these three types of distributions}
    \State Independently draw samples $ X_1, X_2, \ldots, X_{n_i} \sim F_i(x) $;
    \State
    $ \bar{Z}_{n_i} \gets n_i^{-1} \sum_{j=1}^{n_i} x_{j}$; 
    \State $ M_{n_i} \gets \{(n_i-1)^{-1} \sum_{j=1}^{n_i} (x_j - \bar{Z}_{n_i})^2\}^{1/2}$;
\EndIf
\end{algorithmic}
\end{algorithm}

% algorithm version
\begin{algorithm}
\caption{Double Bootstrap percentile with transformed resamples}
\label{al:dbp}
\begin{algorithmic}[1]
\State \textbf{Input:} time $t$, confidence level $1-\alpha$, resampling times $B, C$, system structure $\psi$, the types of components and data $\cup_{i=1}^s\mathcal{X}_i$.
\State \textbf{Output:} Lower confidence limit $R_{DBPT}(t)$.
\For{$i = 1$ \textbf{to} $s$} \Comment{Moment based estimates}
    \State Compute the $\hat{r}_i(t)$ based on $\mathcal{X}_i$ through \eqref{equ:mu and sigma} and \eqref{equ:estimator}. 
\EndFor
\State 
% Calculate the system reliability estimate 
$\hat{R}(t) \gets \psi (\hat{r}_1,\ldots,\hat{r}_s)$.
\State Generate $\bar{Z}_{n_i}^{*j}$ and $M_{n_i}^{*j}$ for $i=1,\ldots,s$ and $j=1,\ldots,B$ independently by Algorithm~\ref{al:type-generating}.

\For{$j = 1$ \textbf{to} $B$} 
\Comment{First layer}
        \For{$i = 1$ \textbf{to} $s$} 
        % \State Compute the $\hat{r}^{*}_{i,j}(t)$ by substituting with $\bar{Z}^{*j}_{n_i},M^{*j}_{n_i}$ and $\hat{r}_i(t)$ in \eqref{equ:MME}. 
        % {\color{red}Give the formula directly.}
        \State $\hat{r}_{i,j}^*(t) \gets  1 - F_i \left[ \left\{F_i^{-1}(1 - \hat{r}_i(t)) -  \bar{Z}_{n_i}^{*j}\right\} \kappa_{i2} / M_{n_i}^{*j} + \kappa_{i1} \right].$
        \EndFor
        \State $\hat{R}_j^*(t) \gets \psi(\hat{r}_{1,j}^*(t), \ldots, \hat{r}_{s,j}^*(t))$.
\EndFor
\State Generate $\bar{Z}_{n_i}^{**l}, M_{n_i}^{**l}$ for $i=1,\ldots,s$ and $l=1,\ldots,C$ independently by Algorithm~\ref{al:type-generating}.
\For{$j = 1$ \textbf{to} $B$}  \Comment{Second layer}
    \For{$l = 1$ \textbf{to} $C$}
        \For{$i = 1$ \textbf{to} $s$}
        %     \State Compute $\hat{r}_{i,j,l}^{**}(t)$ by substituting with $\bar{Z}^{**l}_{n_i},M^{**l}_{n_i}$ and $\hat{r}_{i,j}^{*}(t)$ in \eqref{equ:MME}.
        % {\color{red}Give the formula directly.}
        \State $\hat{r}_{i,j,l}^{**}(t) \gets  1 - F_i \left[ \left\{F_i^{-1}(1 - \hat{r}_{i,j}^*(t)) -  \bar{Z}_{n_i}^{**l}\right\} \kappa_{i2} / M_{n_i}^{**l} + \kappa_{i1} \right].$
        \EndFor
    \State $\hat{R}_{j,l}^{**}(t) \gets \psi(\hat{r}_{1,j,l}^{**}(t), \ldots, \hat{r}_{s,j,l}^{**}(t))$.
\EndFor
\State $\hat{u}_j \gets C^{-1} \sum_{l=1}^C I\left(\hat{R}_{j,l}^{**} \leq \hat{R} \right)$.
\EndFor
%\State Sort as $\hat{u}_{(1)} \leq \hat{u}_{(2)} \leq \ldots \leq \hat{u}_{(B)}$.
\State $R_{DBPT}(t) \gets \hat{R}^*_{(k')}(t)$ where $k = \lceil B\alpha \rceil$ and $k' = \lceil B \hat{u}_{(k)} \rceil$.
\end{algorithmic}
\end{algorithm}

The computational savings of our method come at the cost of introducing dependence among second-layer bootstrap estimates $\hat{r}_{i,j,l}^{**}$. 
Fortunately, this dependence does not affect the asymptotic accuracy of the proposed method, as shown in the following theorem.

\begin{figure}[ht]
\begin{minipage}[b]{0.45\linewidth}
\centering
 \begin{tikzpicture}[scale=0.68, transform shape]
    \node[left] at (0.4,4) {$r_1(t)$};
    \draw [ ->] (0,3.5)--(0,2.5);
    \node[left] at (0.35,2) {$\mathcal{X}_1$};
    \draw [ ->] (0,1.5)--(0,0.5);
    \node[left] at (0.4,0) {$\hat{r}_1(t)$};

    % D1* R1*
    \draw [ ->] (0,-0.5)--(-2,-1.5);   
    \node[left] at (-2, -2) {$\mathcal{X}_{11}^{*}$};
    \draw[->] (-2.5,-2.5)--(-3,-3.5);
    \node[left] at (-2.8,-4) {$\hat{r}^*_{11}$};
    % D11**
    \draw [ ->] (-3.3,-4.5)--(-4,-5.5);
    \node[left] at (-4,-6) {$\mathcal{X}_{111}^{**}$}; 
    \draw[->] (-4.4,-6.5)--(-4.4,-7.5);
    \node[left] at (-4,-8) {$\hat{r}_{111}^{**}$};
    
    \draw [ ->] (-3.3,-4.5)--(-3.3,-5.5);
    \node[left] at (-3,-6) {$\ldots$}; 
    \draw[->] (-3.3,-6.5)--(-3.3,-7.5);
    \node[left] at (-3,-8) {$\ldots$};
    % D1C**
    \draw [ ->] (-3.3,-4.5)--(-2.6,-5.5);
    \node[left] at (-1.8,-6) {$\mathcal{X}_{11C}^{**}$}; 
    \draw[->] (-2.2,-6.5)--(-2.2,-7.5);
    \node[left] at (-1.8,-8) {$\hat{r}_{11C}^{**}$};

    % ... (Example for more R_i*)
    \draw [ ->] (0,-0.5)--(0,-1.5);  % Vertical for central one
    \node[left] at (0.35, -2) {$\ldots$};
    \draw [ ->] (0,-2.5)--(0,-3.5);
    \node[left] at (0.35,-4) {$\ldots$};
    \draw [ ->] (0,-4.5)--(0,-5.5);
    \node[left] at (0.35,-6) {$\ldots$};
    \draw[->] (0,-6.5)--(0,-7.5);
    \node[left] at (0.35,-8) {$\ldots$};

    % DB* RB*
    \draw [ ->] (0,-0.5)--(2,-1.5);   
    \node[left] at (2.6, -2) {$\mathcal{X}_{1B}^{*}$};
    \draw[->] (2.5,-2.5)--(3,-3.5);
    \node[left] at (3.6,-4) {$\hat{r}^*_{1B}$};
    % DB1**
    \draw [ ->] (3.3,-4.5)--(4,-5.5);
    \node[left] at (5,-6) {$\mathcal{X}_{1BC}^{**}$}; 
    \draw[->] (4.4,-6.5)--(4.4,-7.5);
    \node[left] at (5,-8) {$\hat{r}_{1BC}^{**}$};
    
    \draw [ ->] (3.3,-4.5)--(3.3,-5.5);
    \node[left] at (3.6,-6) {$\ldots$}; 
    \draw[->] (3.3,-6.5)--(3.3,-7.5);
    \node[left] at (3.6,-8) {$\ldots$};
    % DBC**
    \draw [ ->] (3.3,-4.5)--(2.6,-5.5);
    \node[left] at (2.7,-6) {$\mathcal{X}_{1B1}^{**}$}; 
    \draw[->] (2.2,-6.5)--(2.2,-7.5);
    \node[left] at (2.7,-8) {$\hat{r}_{1B1}^{**}$};

    \node[left,text width=2cm] at (-4.5,0) {First Layer:};
    \node[left,text width=3cm] at (-4.5,-4) {Second Layer:};
\end{tikzpicture}
\end{minipage}
\hfill
\begin{minipage}[b]{0.4\linewidth}
\centering
    \begin{tikzpicture}[scale=0.68, transform shape] 
    \node[left] at (0.35,4) {$r_1(t)$};
    \draw [ ->] (0,3.5)--(0,2.5);
    \node[left] at (0.35,2) {$\mathcal{X}_1$};
    \draw [ ->] (0,1.5)--(0,0.5);
    \node[left] at (0.35,0) {$\hat{r}_1(t)$};

    % D1* R1*
    \draw [ ->] (0,-0.5)--(-2,-1.5);   
    \node[left] at (-1.6, -2) {$\bar{Z}^{*1}_{n_1},M^{*1}_{n_1}$};
    \draw[->] (-2.5,-2.5)--(-3,-3.5);
    \node[left] at (-2.8,-4) {$\hat{r}^*_{11}$};
    \draw[->] (-3.2,-4.5)--(-0.5,-5.5);
 
    % ... (Example for more R_i*)
    \draw [ ->] (0,-0.5)--(0,-1.5);  % Vertical for central one
    \node[left] at (0.35, -2) {$\ldots$};
    \draw [ ->] (0,-2.5)--(0,-3.5);
    \node[left] at (0.35,-4) {$\ldots$};
    \draw [ ->] (0,-4.5)--(0,-5.5);
    \node[left] at (2.3,-6) {$\bar{Z}^{**l}_{n_1},M^{**l}_{n_1},l=1,\ldots,C$};

    % DB* RB*
    \draw [ ->] (0,-0.5)--(2,-1.5);   
    \node[left] at (3.4, -2) {$\bar{Z}^{*B}_{n_1},M^{*B}_{n_1}$};
    \draw[->] (2.5,-2.5)--(3,-3.5);
    \node[left] at (3.4,-4) {$\hat{r}^*_{1B}$};
    \draw[->] (3,-4.5)--(0.5,-5.5);

    % second 
    \draw[->] (0,-6.5)--(0,-7.5);
    \node[left] at (0.35,-8) {$\ldots$};

    \draw[->] (0.5,-6.5)--(3,-7.5);
    \node[left] at (4.5,-8) {$\hat{r}_{1B1}^{**},\ldots,\hat{r}_{1BC}^{**}$};
    
    \draw[->] (-0.5,-6.5)--(-3.2,-7.5);
    \node[left] at (-2.1,-8) {$\hat{r}_{111}^{**},\ldots,\hat{r}_{11C}^{**}$};
\end{tikzpicture}
\end{minipage}
\caption{Illustration of the conventional double bootstrap percentile method (left) compared with our double bootstrap percentile method with transformed resamples (right) applied for assessing the reliability of the first component, $r_1(t)$.}
\label{fig:tik}
\end{figure}
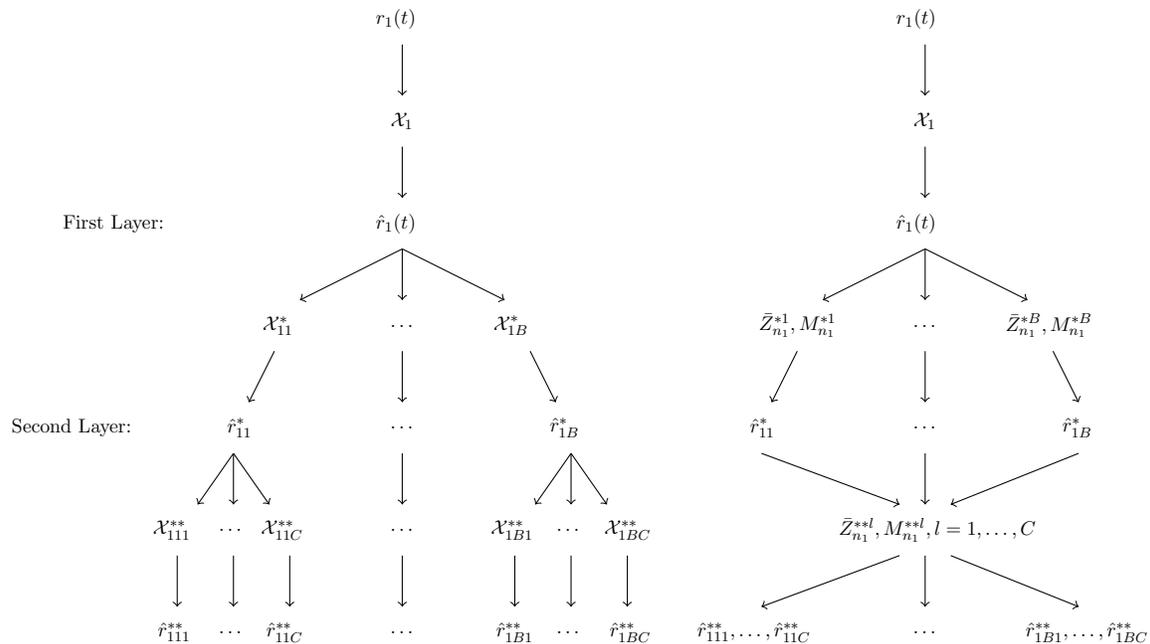

\begin{theorem}
\label{thm:order of LCL}
Assume the distribution $G_{\theta_i}(t)$ is continuous and exists at least fourth-order moments for $i=1,\ldots,s$. 
The lower confidence limit obtained by Algorithm~\ref{al:dbp}, $R_{DBPT}(t)$, is within $[0,1]$ and the coverage level is
\begin{equation}
% \lim_{B,C \rightarrow \infty}
\text{P} \left( R(t) \geq R_{DBPT} \right) = 1- \alpha + O(n^{-1}).
% + O(B^{-1/2} + C^{-1}).
\end{equation}
\end{theorem}

% Theorem
The proof is provided in the supplementary materials.
In practical applications, the finite-sample performance of the double bootstrap percentile method depends on the choice of parameters $B$ and $C$. 
For SRA with $n \leq 10^2$, we recommend set $ B = 10^3 $ and $ C = 500 $.

In conclusion, our method shows significant advancements in accuracy and efficiency compared to the state-of-the-art methods.
Firstly, it achieves second-order asymptotic accuracy for the LCL, matching the theoretical precision of CF expansion-based approaches without requiring complex analytical derivations.
Second, compared to the conventional double bootstrap percentile method, we reduce the computational complexity from $O(snBC)$ to $O(s(B + C) + sBC)$.
Thirdly, the proposed method ensures that LCLs remain within the valid interval $[0,1]$, and further numerical results in Section~\ref{sec:numerical_studies} show that $R_{DBPT}$ rarely exhibits bend-back issues.

\section{Numerical Studies}
\label{sec:numerical_studies}

In this section, we evaluate the performance of our double bootstrap percentile method with transformed resampling through a series of numerical simulations. 
For comparison purposes, the following alternative methods are considered: 
\begin{itemize}
    \item The \textbf{Delta method} \citep{hong2014confidence}, a widely used first-order asymptotic approach; 
    \item The \textbf{R-WCF method} \citep{li2020higher}, which achieves higher-order convergence than the Delta method; 
    \item The \textbf{conventional bootstrap percentile method (BP)} as introduced in Section 3.2; 
    \item Our proposed \textbf{double bootstrap percentile method with transformed resamples (DBPT)}. 
\end{itemize}
Both complete and censored data scenarios are examined separately in Section \ref{subsec:complete} and \ref{subsec:censored}, respectively. 

The following evaluation metrics are employed:
\begin{itemize}
    \item The \textbf{empirical coverage probability} is a primary metric for evaluating confidence limits, with closer alignment to the nominal confidence level indicating better accuracy.  
    \item The \textbf{quantiles of the LCLs} obtained by repeated evaluation of the alternative methods—benchmarks for assessing the accuracy of LCLs as in \cite{li2020higher}. 
    This evaluation metric is also a widely used metric, with values closer to the true reliability indicating better performance.
\end{itemize}
Besides these quantitative evaluation metrics, we also assess how our method addresses the common “falling outside” and “bend-back” issues, which are frequently observed in other methods. 
Finally, we will compare the computational costs of our method to those of the conventional double bootstrap percentile method.
Note that throughout this section, we set the confidence level at $1-\alpha = 0.9$.

\subsection{Simulation for complete data}
\label{subsec:complete}
In our initial example, we consider a parallel system and a series system, each comprising three Weibull-type components. 
Such configurations are commonly encountered in reliability analyses: series systems often model competing-risk scenarios, while parallel systems represent redundant designs.
We set the true reliability to $ R = 0.9988 $ for the parallel system and $ R = 0.9548 $ for the series system, and we construct LCLs at a confidence level of $ 1-\alpha = 0.9 $. 
For both BP and DBPT, we set $B = 10^3$, and for DBPT, we set $C = 500$.
For simplicity, we assume identical sample sizes $ n_i = n $ for $ i = 1, 2, 3 $. 
However, our method is readily applicable to situations with unequal sample sizes without any modification. 
We repeat the experiment $ 10^4 $ times to assess the empirical coverage probability, and the $90\%$ quantile of the LCLs as the component sample size $n$ varies. 

As illustrated in Fig.~\ref{fig:two system}, the DBPT method achieves the most accurate coverage probabilities, precisely matching the nominal level once $n \geq 10$. 
The R-WCF method converges more rapidly to the nominal level than the Delta and BP methods, reflecting its second-order correctness, whereas both the Delta and BP methods exhibit only first-order correctness.
Moreover, the DBPT method’s $90\%$ quantile closely aligns with the true reliability, significantly outperforming the other methods for both the parallel and series systems.
    \begin{figure}[!t]
        \centering
        \includegraphics[width = 0.475\textwidth]{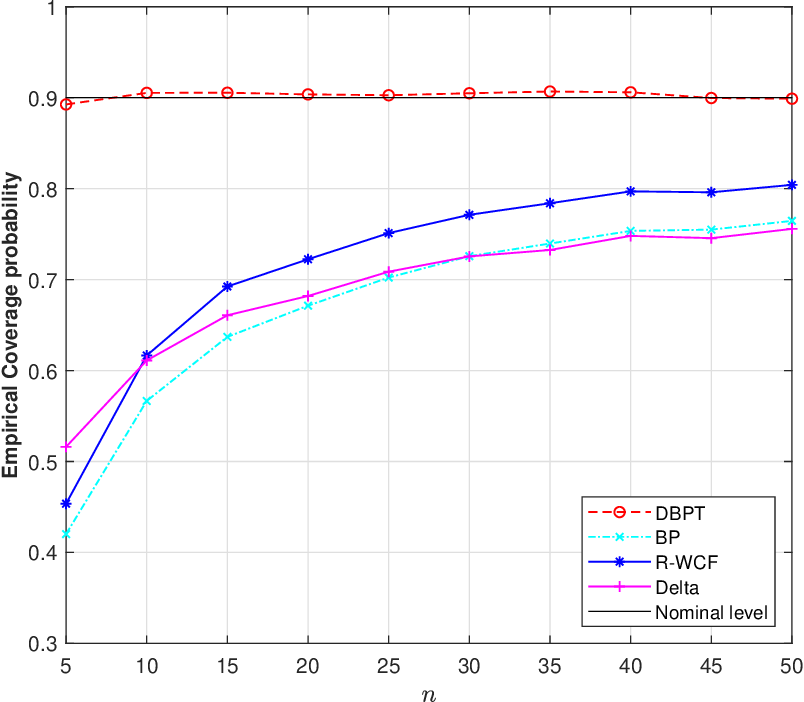}
        \hfill
        \includegraphics[width = 0.48\textwidth]{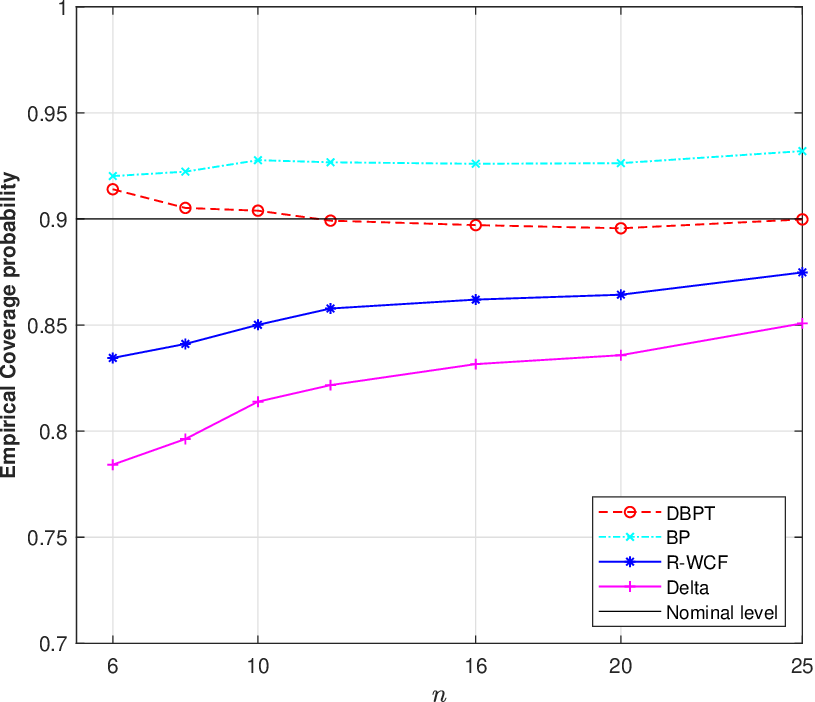}
        \\[0.8cm]
        \includegraphics[width = 0.48\textwidth]{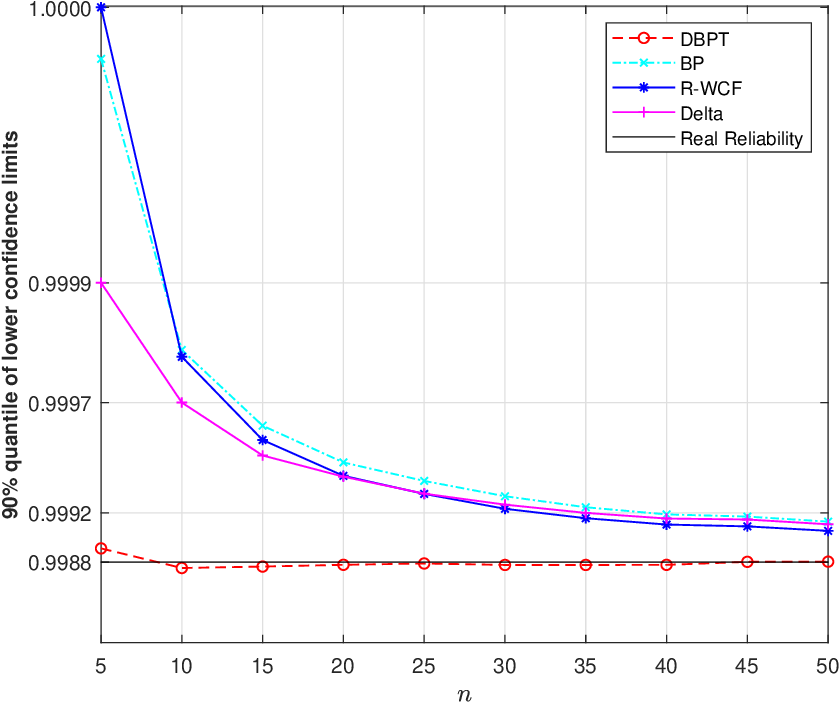}
        \hfill
        \includegraphics[width = 0.48\textwidth]{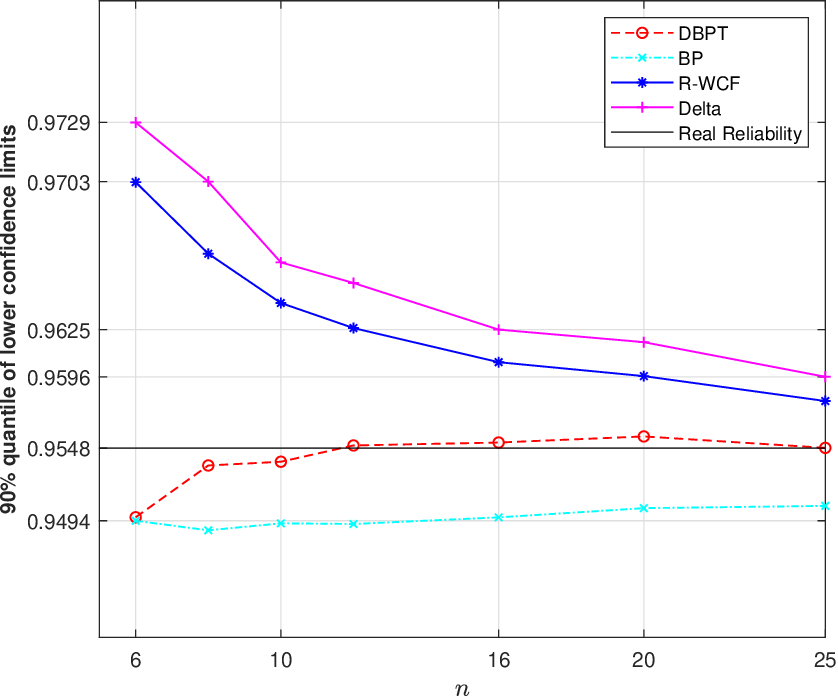}
        \caption{The empirical coverage probability (top)
        % , the mean (middle) 
        and the $90\%$ quantile (bottom) of lower confidence limits versus $n$ for a parallel (left) and a series (right) system with $s = 3$ independent Weibull components and $1-\alpha = 90\%$.}
        % with true reliability $R=0.9988$ for parallel system and $R=0.9548$ for series system. The experiment was repeated $10^4$ times, with $B=10^3$ and $C=500$ for DBPT method. }
        \label{fig:two system}
    \end{figure}

To further validate the accuracy of our method, we evaluate its performance on two more challenging systems. 
First, we consider a $2\times 2$ series-parallel system composed of $s=4$ log-normal components. 
Second, we consider a 5-out-of-8 system \citep{koutofn2017} composed of $s=8$ Weibull components, where the system is functional if at least 5 components are functional. 
As shown in Fig.~\ref{fig:other-system}, our method consistently outperforms other approaches in both scenarios.
These results demonstrate that our method performs robustly across different system structures and lifetime models. 
        \begin{figure}
        \centering
        \includegraphics[width = 0.47\textwidth]{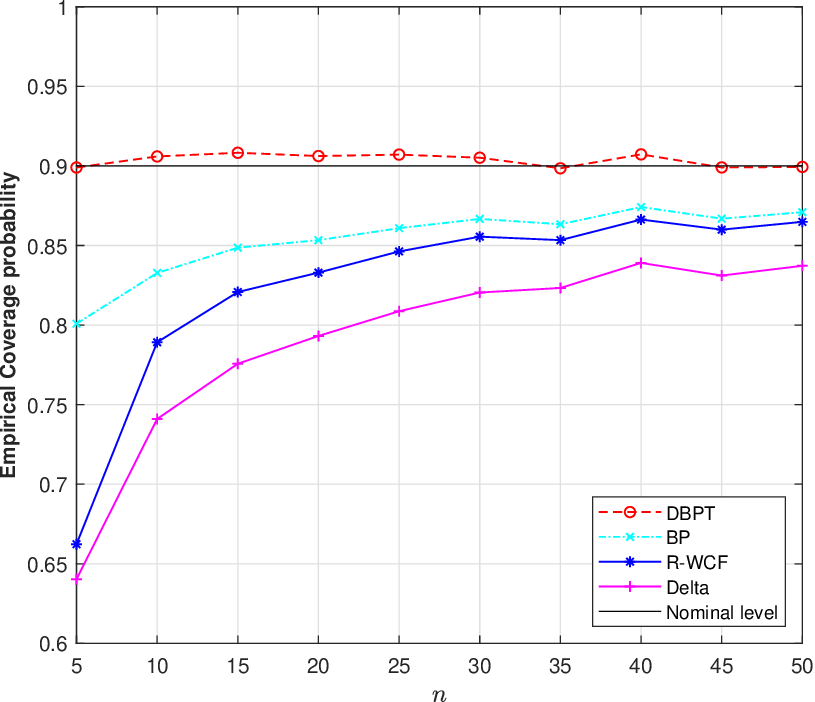}
        \hfill
        \includegraphics[width = 0.46\textwidth]{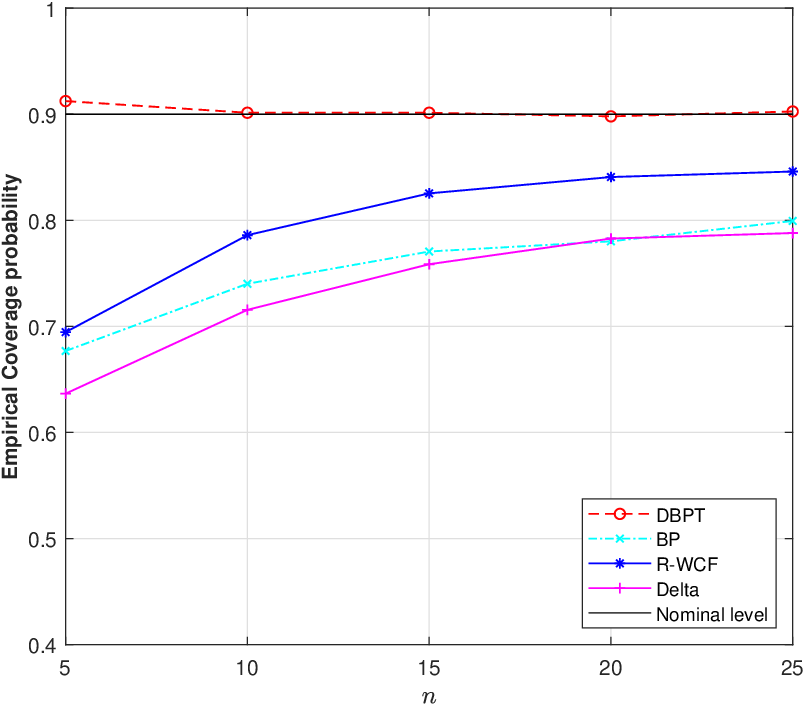}
        \\[0.6cm]
        \includegraphics[width = 0.47\textwidth]{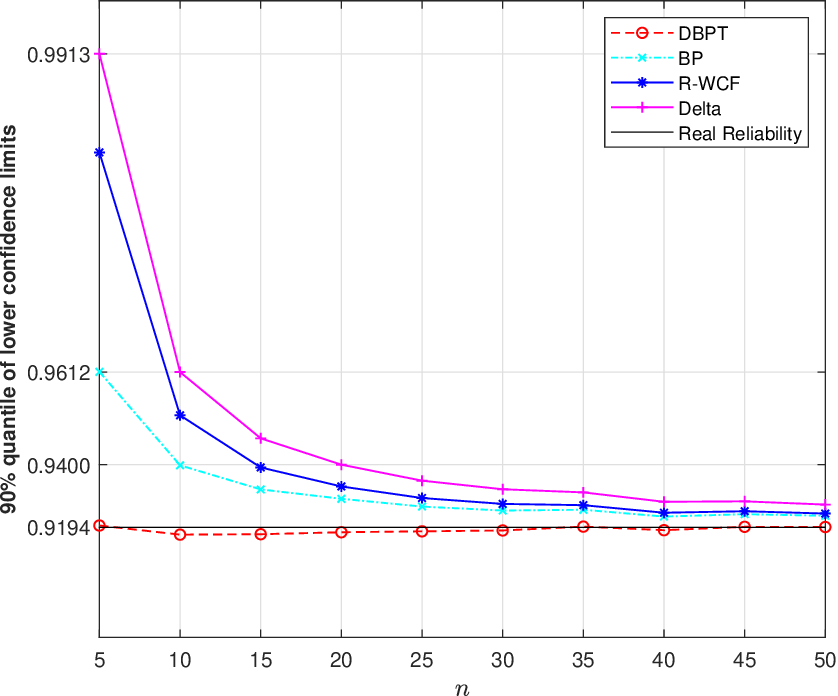}
        \hfill
        \includegraphics[width = 0.46\textwidth]{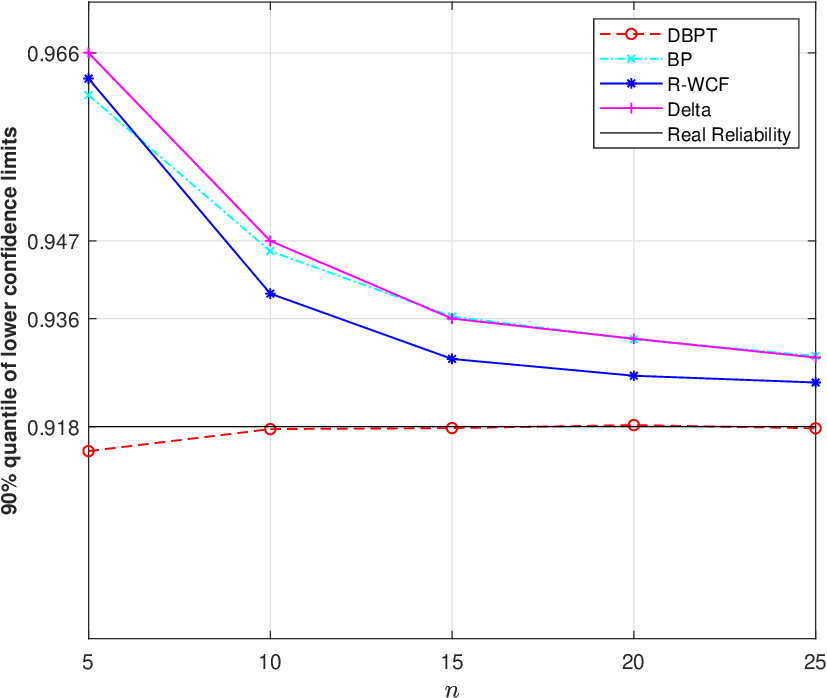}
        \caption{The empirical coverage probability (top)
        and the $90\%$ quantile (bottom) of lower confidence limits versus $n$ for a $2\times 2$ series-parallel system (left) and a 5-out-of-8 system (right).}
        \label{fig:other-system}
    \end{figure}

In the second example, we evaluate the performance of our method as the true reliability $ R $ varies. 
We let the $R$ range from $0.5$ to $0.999999$, with the confidence level of $1-\alpha = 0.9$. 
High-reliability assessment is always challenging in SRA, so we include some cases with $R$ close to 1 to challenge the proposed method. 
The analysis considers both a parallel system and a 5-out-of-8 system, each evaluated under two distinct sample sizes, $n = 10$ and $n=50$.
As illustrated in Fig.~\ref{fig:reliability}, our method remains stable across a range of $R$-values for both $n=10$ and $n=50$.
By contrast, when $ R $ approaches  $1$, the empirical coverage probabilities of other methods increasingly deviate from the confidence level. 
Although increasing the sample size to $n=50$ yields some improvements, it seems that this is still insufficient for alternative methods in high-reliability regimes as their coverage probability are far from the nominal confidence level.
The primary reason is that the variance of the point estimator grows rapidly in these regions, often pushing the confidence limits outside the interval $[0,1]$. 
This issue is notably severe for the R-WCF methods, which also frequently encounter numerical challenges.

\begin{figure}[!t]
        \centering
        \includegraphics[width = 0.48\textwidth]{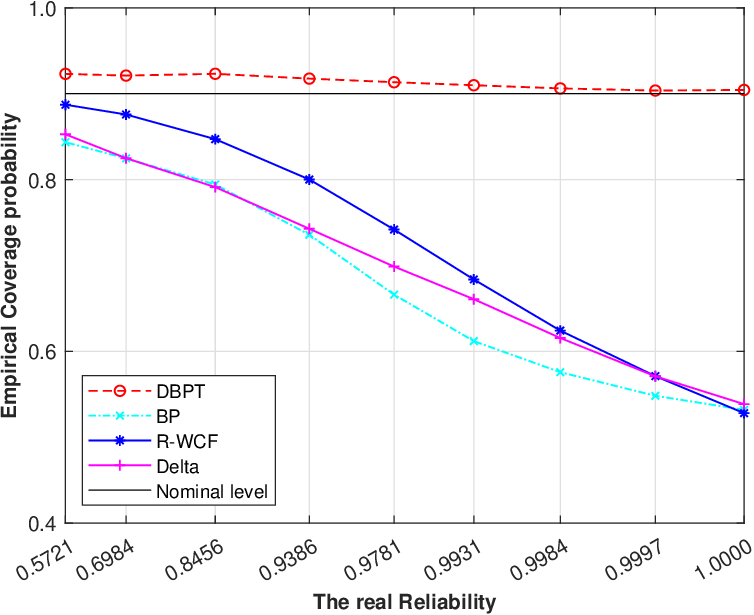}
        \hfill
        \includegraphics[width = 0.48\textwidth]{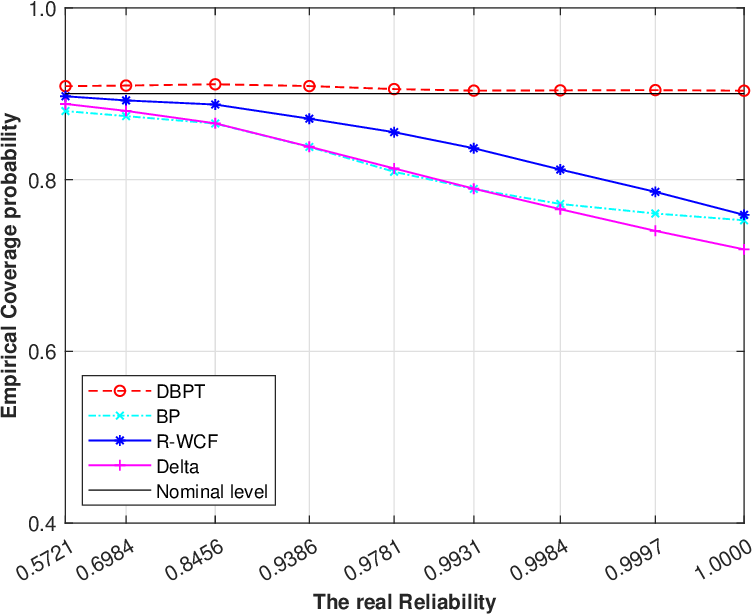}
        \\[0.8cm]
        \includegraphics[width = 0.48\textwidth]{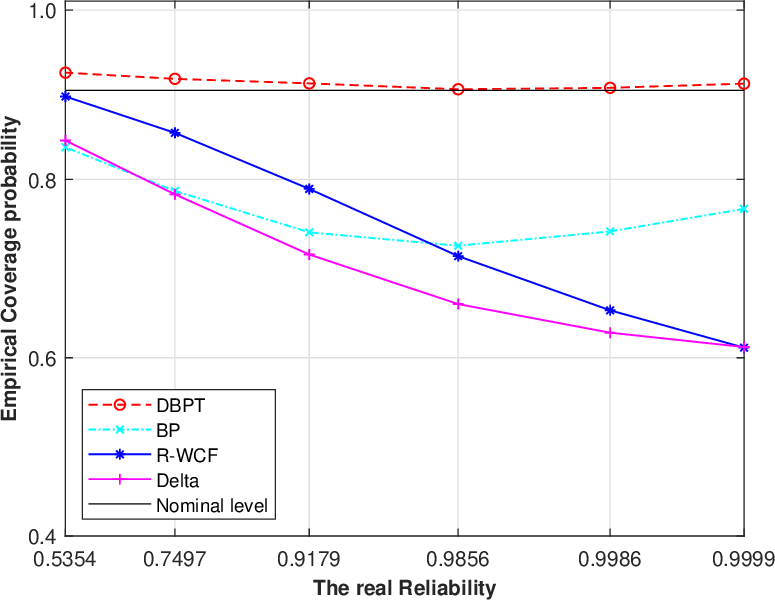}
        \hfill
        \includegraphics[width = 0.48\textwidth]{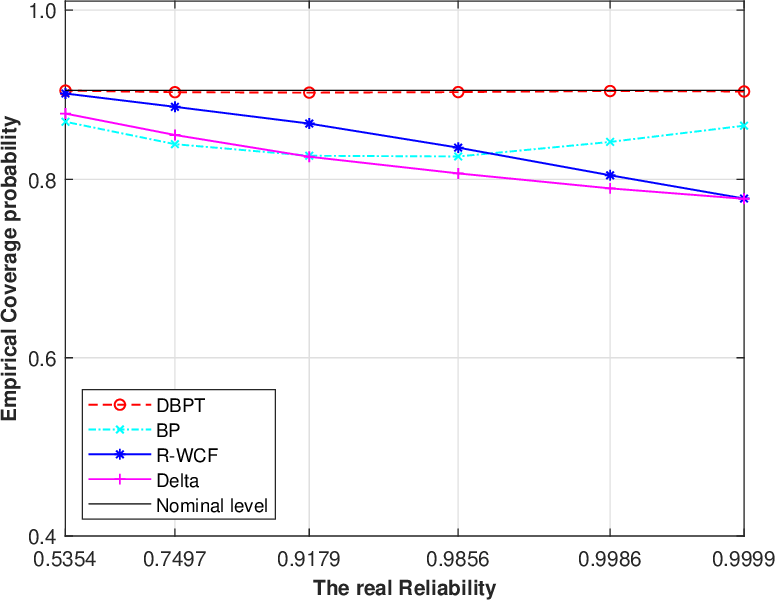}
        \caption{The empirical coverage rate versus the true reliability for a parallel (top) and a 5-out-of-8 (bottom) system with sample size $n=10$ (left) and $n=50$ (right) and $1 - \alpha = 90\%$.}
    \label{fig:reliability}
    \end{figure}

     \begin{table}
        \centering
        \resizebox{\textwidth}{!}{  
        \begin{tabular}{l|c|c|c|c}
            \textbf{Method} & \textbf{Error} & \textbf{Falling Outside} & \textbf{Bend Back} & \textbf{Numerical Stability} \\
            \hline
            Delta & $O(n^{-1/2})$ & 8545 & 1038 & \checkmark \\
            RWCF & $O(n^{-1})$ & 1714 & 906 & $\times$ (Ineffective for boundary)  \\
            BP & $O(n^{-1/2})$ & 0 & 0 & \checkmark \\
            DBPT & $O(n^{-1})$ & 0 & 20 & \checkmark \\
        \end{tabular}
        }
        \caption{Statistical methods evaluated for error rate, outliers, bend-back frequency, and numerical stability, where $n=5, \alpha =0.1$, repeated $10^4$ times.}
        \label{tab:comparison_results}
    \end{table}

We further examine the ability of avoiding “falling outside” and “bend-back” by quantifying the probabilities of these issues observed in simulation results.
After conducting $10^4$ experimental runs for a series system under $ n = 5 $ and a nominal confidence level of $0.9$, we recorded the number of times these issues arose, relating these occurrences to the theoretical appearance probability of the methods.
As shown in Table~\ref{tab:comparison_results}, the DBPT method yields LCLs that remain within the valid range and seldom exhibit bend-back behavior, with only $20$ occurrences in $10^4$ experiments. 
Although the BP method never experiences bend-back issues, it does not reach the precise coverage accuracy achieved by the DBPT method. In contrast, other methods frequently produce confidence limits that fall outside the allowed interval or display bend-back problems. Notably, the R-WCF method frequently encounters numerical issues when the true reliability is close to 0 or 1.

Regarding computational efficiency, we consider a 9-out-of-16 system, which contains 16 components, and is functional if at least 9 components operate.
The sample size for each component is $n=100$. 
Remarkably, a single experiment using our improved double bootstrap percentile method required only $0.2$ hours, while the conventional double bootstrap percentile method needed $78.8$ hours on the same computer.\footnote{It was performed on a laptop with an AMD Ryzen 7 5800H (8-cores) and 16 GB DDR4-3200MHz RAM.}
Moreover, the LCLs derived from both methods show excellent numerical agreement, matching to the first two decimal places.
% \lzhcmt{How about the results? You can report shortly.}
By incorporating moment-based estimation and the transformed resamples technique, we significantly improve computational efficiency without compromising the accuracy of the confidence limits.
This advancement makes our method more practical for complex system reliability assessment.

Overall, our method exhibits the best performance for complete data, with the empirical coverage probability closely matching the nominal confidence level.
Moreover, it effectively avoids the issue of confidence limits falling outside the valid range and rarely encounters the bend-back issue. 
In addition, our method reduces runtime by orders of magnitude compared to traditional double bootstrap approaches.

\subsection{Simulation for censored data}
\label{subsec:censored}
In real-world applications, censored data is commonly encountered due to limited experimental resources \citep{meeker2022statistical,CensorSchneider,xiao2014study}.
We now conduct a simulation to evaluate the effectiveness of our proposed method in the presence of censored data. 
For consistency, we adopt the same parallel system configuration and parameter settings as detailed in Section~\ref{subsec:complete}.
We consider the Type II right-censoring mechanism \citep{CensorSchneider}, which is commonly used in reliability experiments and involves terminating testing after a predefined number of failures.
Let $ t_{(1)} \leq t_{(2)} \leq \dots \leq  t_{(n)} $ denote the ordered (potential) failure times. 
In this study, the experiment is terminated once $ \tilde{n} = \lceil 0.7n \rceil $ failures occur.  
The collected data consists pairs $(x_i,\delta_i)$ where $x_i = t_{(i)}, \delta_i = 0$ if $i \leq \tilde{n}$ and otherwise $x_i = t_{\tilde{n}}, \delta_i = 1$.
Here $\delta_i$ serves as the censoring indicator.

To handle the censored data, we apply the data imputation algorithm proposed by \cite{xiao2014study} to obtain the virtually complete data and then apply alternative methods (except the Delta method, which is capable of handling the censored data directly) to calculate the LCLs of system reliability.
For the delta method, we employ the EM algorithm to obtain the MLE from the censored data, and the Fisher's information is calculated from the likelihood of the censored data.
% However, it is hard to calculate the Fisher’s information using censored data, so the asymptotic variance of delta method is calculated using the filled data.

    \begin{figure}[!t]
        \centering
        \includegraphics[width = 0.46\textwidth]{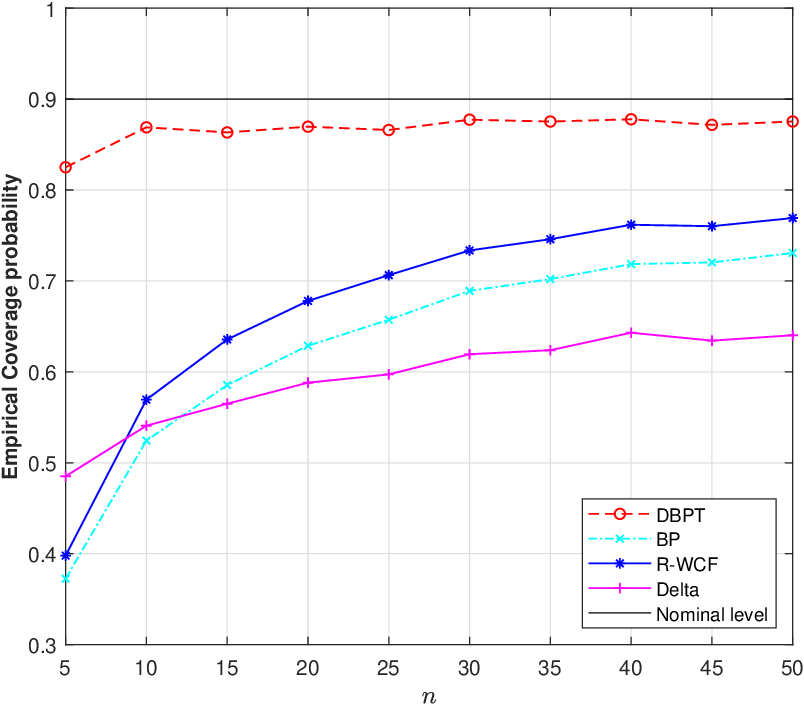}
        \hfill
        \includegraphics[width = 0.48\textwidth]{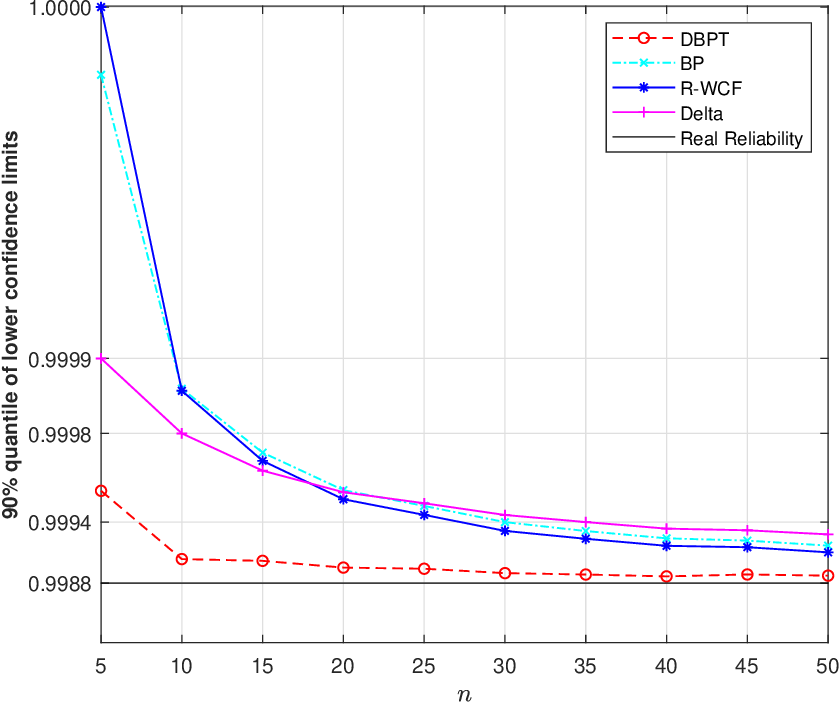}
        \hfill
        \caption{The empirical coverage probability (left) and the $90\%$ quantile (right) of lower confidence limits versus $n$ for a parallel system composed of $s = 3$ independent Weibull components with $30\%$ censored data and $1-\alpha = 90\%$.}
        \label{fig:PWcensor}
    \end{figure}

The results are shown in Fig.~\ref{fig:PWcensor}.
These results demonstrate significant advantages of the proposed method. 
Even with a small sample size, the empirical coverage probabilities are close to the nominal confidence level, and the $90\%$ quantiles of the LCLs align closely with the true reliability. 
Compared to the complete data scenario, the empirical coverage probability shows slight deviation from the nominal confidence level for all methods. 
This is primarily due to the loss of information caused by data censoring. 
Nonetheless, our method consistently outperforms the alternatives across all metrics in the censored data case.

% {\color{purple}My overall feeling about Section 4 is as follows: 1, we may consider to add the supplementary material results on 2*2 system to the main text, if page limit allows. 2, n=10 and n=5 may be too small? Although they may be the best in showing our advantage, we can still add results like n=50, as long as our method is not worse than others. 3, I think we have presented the right criteria for comparing methods and the right competitors. And our methods are much better than others. That is good. 4, Maybe we show some results with much larger n (and consequently much larger B and C). This may show the necessity of using the transformed resampling technique. For now, although our method is much faster than conventional DBP, conventional DBP is fast enough. }

\section{Conclusions}
\label{sec:conclusions}
In this paper, we propose a novel double bootstrap percentile framework for system reliability assessment. 
By leveraging moment-based estimators for log-location-scale distributions, we invent an innovative resampling strategy that eliminates the second-stage resampling and estimator recalibration required by conventional double bootstrap methods, thus reducing computational demands. 
Furthermore, the method is highly flexible, allowing it to effectively handle varying component data sizes and system structures. 
For censored data, a common challenge in system reliability analysis, we integrate an established data imputation algorithm, ensuring robust performance with incomplete observations.
Our approach achieves high-order, accurate lower confidence limits without complex mathematical expansion formulas. 
Numerical studies show that our method exhibits the best performance in terms of the coverage probability.
We also show that the proposal method can drastically reduce the appearance of the bend-back issue compared to other high-order asymptotic SRA methods. 

First, according to the simulation results, our method cannot completely eliminate the bend-back issue. 
This is a minor compromise made in order to achieve a substantial improvement in accuracy compared to the bootstrap percentile method. 
In the future, efforts will be made to fully address this issue.
Second, in some high-reliability scenarios, zero-failure component data may be encountered. 
How to modify our method to handle this situation will be the focus of our future work.

Moreover, two promising research directions will be considered in future work: first, extending the proposed method to handle degradation data for reliability assessment is a vital task in modern reliability engineering. 
Second, in many modern engineering problems, lifetime data are typically collected with numerous covariates.
Building reliability models, e.g., regression models, that incorporate important covariates is essential to improve the quality of the reliability assessment.
How to extend our proposed acceleration principle in regression analysis is an important future research topic.

\bibliographystyle{apalike}
\spacingset{1}
\bibliography{refs}
\appendix

\section{proofs}

% \begin{proof}[Proof of the \eqref{equ:example_exponential} in Example \ref{exam:bayesian}:]
% The coverage probability is 
% \begin{align}
% \label{equ:pf example exp}
% \text{P}\left(r(t) > e^{-t q(\alpha, n+1) / \sum_{j=1}^n t_j}\right) 
% &= \text{P}\left(\lambda \sum_{j=1}^n t_j < q(\alpha, n+1)\right). 
% \end{align}
% Since each $t_j$ independently follows an exponential distribution, it follows that $\lambda t_j$ follows a $\text{Gamma}(1, 1)$ distribution. Consequently, $\lambda \sum_{j=1}^n t_j$ follows a $\text{Gamma}(n, 1)$ distribution. 
% Denote $F_{n,1}(x)$ and $f_{n,1}(x)$ as the cdf and pdf of the Gamma$(n,1)$, respectively.
% Continuing from Eq.\eqref{equ:pf example exp}, we have
% \begin{align*}
% \text{P}\left(\lambda \sum_{j=1}^n t_j < q(\alpha, n+1)\right) 
% % &= \int_{0}^{q(\alpha, n+1)} f_{n, 1}(x) \mathrm{d}x \\ 
% &= F_{n, 1}(q(\alpha, n+1) )   \\ 
% % &= \int_0^{q(\alpha, n+1)} e^{-x}/{\Gamma(n+1)} \mathrm{d} x^n \\
% &= F_{\Gamma(n+1, 1)}(q(\alpha, n+1)) + f_{n+1,1}(q(\alpha, n+1)) \\ 
% &= 1 - \alpha + O(n^{-1/2}). 
% \end{align*}
% The second equality holds by integration by parts, and the final equality follows from Stirling's approximation.
% \end{proof}

\begin{lemma}\label{lem:psi}
For a coherent system consisting of $k$ components, the system reliability $R(t)$ can be represented by a multivariate polynomial w.r.t. component reliability $r_1,r_2,...,r_k$.
\end{lemma}
\begin{proof}
Without loss of generality, suppose $t$ is fixed. 
Let $S_i$ denote the state of failure of individual component $i$; $S_i=0$ represents the failure of the $i$-th component at time $t$.
Then we  have $P(S_i = 1) = r_i(t)$, $S_i,S_j$ are independent for $i\neq j$.
We now introduce the notation $\delta(s_1,s_2,...,s_k)$, which represents the state of failure for the system given the components' state being $(s_1,s_2,...,s_k)$.
According to the system structure, $\delta(s_1,s_2,...,s_k)$ is either $0$ (fail) or $1$ (work) and is deterministic.  
For a given system, we can write that $R(t) = \sum_{(s_1,...,s_k)\in\{0,1\}^k} P(S_i=s_i,i=1,2,...,k)\delta(s_1,s_2,...,s_k) = \sum_{(s_1,...,s_k)\in\{0,1\}^k} r_i(t)^{s_i}\delta(s_1,s_2,...,s_k) = \sum_{\delta(s_1,s_2,...,s_k)=1} r_i(t)^{s_i}$.
    It is easy to see that $R(t)$ is a polynomial w.r.t. $r_1,r_2,...,r_k$.
\end{proof}

\begin{proof}[Proof of Theorem~\ref{thm:BP and BB}]

It is easy to show that $\hat{R}^*_{j}$'s take values within the interval $[0,1]$ and are monotonically decreasing w.r.t. $t$. 
Moreover, $R_{BP} = \hat{R}^*_{(\lceil B\alpha \rceil)}(t)$ is the $ \lceil B\alpha \rceil $-th smallest values of $\hat{R}^*_{1},\ldots,\hat{R}^*_{B}$, thus $R_{BP}$ fall within the interval $[0,1]$.

Next we prove that $R_{BP}$ is a monotonically decreasing function w.r.t. $t$.
Let $ t_1 \leq t_2 $, and consider the set
$ A = \{ b : R_b(t_1) \leq R_{\lceil B\alpha \rceil}(t_1) \}$.
Clearly, the cardinality satisfies $ |A| \geq \lceil B\alpha \rceil $. Since each $ R_b(t) $ is a non-increasing function of $ t $, it follows that for all $ b \in A $,
\[
R_b(t_2) \leq R_b(t_1) \leq R_{\lceil B\alpha \rceil}(t_1).
\]
Thus there are at least $ \lceil B\alpha \rceil $ values in the set $\{R_b(t_2), b = 1,\ldots,B\}$ that less than or equal to $ R_{\lceil B\alpha \rceil}(t_1) $. Therefore, the $ \lceil B\alpha \rceil $-th smallest value in $\{R_b(t_2), b = 1,\ldots,B\}$ satisfies
\[
R_{\lceil B\alpha \rceil}(t_2) \leq  R_{\lceil B\alpha \rceil}(t_1).
\]
This establishes that $ R_{\lceil B\alpha \rceil}(t) $ is a monotonically decreasing function of $ t $.

Finally, we demonstrate that the asymptotic coverage probability derived in this theorem aligns with the cases presented on page 510 of \cite{shao2008mathematical}.
By Lemma~\ref{lem:psi}, the function $R(t) = \psi(r_1,\ldots,r_s)$ is a polynomial and is four times continuously differentiable with bounded derivatives. Furthermore, since $G_{\theta_i}$ is absolutely continuous and thus has an absolutely continuous probability density function. By the Riemann–Lebesgue lemma, the characteristic function of the distribution $G_{\theta_i}(t)$ satisfies
$$\lim_{t \rightarrow \infty}\sup |\chi(t)| = 0.$$ 
Additionally, we have assumed that the distribution of $G_{\theta_i}(x)$ admits finite moments of at least fourth order. Combining these properties, the distribution of $S_n = \sqrt{n}(\hat{R}-R)$ admits the Edgeworth expansion and satisfies all regularity conditions outlined on page 510 of \cite{shao2008mathematical}. 
It follows that both $R_{BB}$ and $R_{BP}$ are first-order asymptotically correct.
\end{proof}

\begin{proof}[Proof of Theorem~\ref{thm:order of DBP}]
Under assumptions analogous to those in Theorem~\ref{thm:BP and BB}, all regularity conditions are satisfied.  
Thus, the asymptotic property posted in this theorem is a special case of Theorem 1 in \citep{martin}.
\end{proof}

\begin{proof}[Proof of Theorem~\ref{thm:mme}:]
% \textbf{Moment-Based Estimation}  
Let $ z_{i,j} = (x_{i,j} - \mu_i)/\sigma_i $ denote the standardized sample.
$ z_{i,j}, j=1,...,n_i$ are then i.i.d. samples from the distribution $F_i(x)$, for $i=1,2,...,s$. 
The sample mean and variance for $z_{i,j} , j=1,2,...,n_i$ are 
$$\bar{Z}_{n_i} = \sum_{j=1}^{n_i} z_{i,j},\quad M_{n_i}^2 = (n_i-1)^{-1}\sum_{j=1}^{n_i} (z_{i,j} - \bar{Z}_{n_i})^2,$$
and thus $\bar{Z}_{n_i} = (\bar{x}_i - \mu_i)/\sigma_i, \quad M_{n_i} = s_i/\sigma_i$, where $\bar{x}_i $ and $s_i$ are sample mean and variance of $x_{i,j} , j=1,2,...,n_i$.
From Eqs.~\eqref{equ:mu and sigma} and \eqref{equ:estimator}, we have
\[
 \hat{r}_i(t) =  1 - F_i \left( (\log t - \hat{\mu}_i)/\hat{\sigma}_i \right) = 1 - F_i \left( \left( \frac{\log t - \mu_i}{\sigma_i} - \bar{Z}_{n_i} \right) M_{n_i}^{-1} \kappa_{i2} + \kappa_{i1} \right).
\]  
Since  $
    r_i(t) = 1 - F_i \left( (\log t - \mu_i)/{\sigma_i} \right)$,  it follows that  
\[
        \hat{r}_i(t) =  1 - F_i \left[ \left\{ F_i^{-1} \left( 1 - r_i(t) \right) - \bar{Z}_{n_i} \right\} M_{n_i}^{-1} \kappa_{i2} + \kappa_{i1} \right].
\]  
This demonstrates that the distribution of $ \hat{r}_i(t) $ is identical to the distribution of  
\[
1 - F_i \left[ \left\{ F_i^{-1} \left( 1 - r_i(t) \right) - \bar{Z}_{n_i} \right\} M_{n_i}^{-1} \kappa_{i2} + \kappa_{i1} \right],
\] 
where $\bar{Z}_{n_i}$ and $M_{n_i}^2$ denote the sample average and the sample variance of $n_i$ samples drawn from the distribution $F_i(x)$.
This proves \eqref{equ:MME}. 
% A similar argument leads that \eqref{equ:MME-BP} holds.

In the following, we consider the Weibull, log-normal, and Exponential distributed lifetimes as examples and develop the corresponding further formulas.\\
\textbf{1) Weibull} type component, with $F_i(x) = F_e(x) = 1-e^{-e^x}$,  the population mean and variance are $\kappa_{i1} =  -\gamma, \kappa_{i2} = \pi/\sqrt{6}$, where $\gamma \approx 0.5772$. Thus we have  
\begin{equation*}
\hat{r}_i(t) \sim_{\text{i.i.d}}  \exp\left( -\exp\left( \frac{\log (-\log r_i(t)) - \bar{Z}_{n_i}}{M_{n_i}} \cdot \frac{\pi}{\sqrt{6}} - \gamma \right) \right).
\end{equation*}
\textbf{2) Log-normal} type component, with $ F_i(x) = \Phi(x) $, we have $ \kappa_{i1} = 0 $ and $ \kappa_{i2} = 1 $.  
Thus we have
\begin{equation*}
\hat{r}_i(t) = 1 - \Phi \left\{ (\Phi^{-1} \left( 1 - r_i(t)\right) - \bar{Z}_{n_i})/M_{n_i}  \right\} 
= \Phi \left\{    (\Phi^{-1}(r_i(t)) + \bar{Z}_{n_i})/M_{n_i} \right\}.
\end{equation*}
Since $ \bar{Z}_{n_i} $ and $ M_{n_i}^2 $ are the sample mean and sample variance of $ n_i $ samples from a normal distribution, we have that $ \bar{Z}_{n_i}/M_{n_i} $ follows a Student's t-distribution $ t(n_i - 1) $, and $ (n_i - 1) M_{n_i}^2 $ follows a chi-squared distribution $ \chi^2(n_i - 1) $.\\  
% Thus, we can express $ \hat{r}_i(t) $ as
% \begin{equation}
% \hat{r}_i(t) \sim_{\text{i.i.d.}} \Phi \left\{ T_{n_i} + \sqrt{(n_i-1)/S_{n_i}} \Phi^{-1}(r_i(t)) \right\},
% \end{equation}
% where $ T_{n_i} $ follows a Student's t-distribution $ t(n_i - 1) $ and $ S_{n_i} $ follows a chi-squared distribution $ \chi^2(n_i - 1) $, and they are independent.\\
\textbf{3) Exponential} type component with the distribution function $P(T_i \leq t) = 1 - \text{e}^{-\lambda_i t}$ is a special case.
We can directly obtain the moment estimator for $\lambda_i$ from the fact $\mathbb{E}(t_{i1}) = 1/\lambda_i$, yielding $\hat{\lambda}_i = n_i /  \sum_{j=1}^{n_i} t_{i,j} $. 
And the corresponding moment-based estimator 
\begin{equation*}
\hat{r}_i(t) = \exp\left( -\hat{\lambda}_i t \right)
= \exp\left( \hat{\lambda}_i \log r_i(t)/\lambda_i\right).
\end{equation*}
Since $\lambda_i / \hat{\lambda}_i = \lambda_i \sum_{j=1}^{n_i}  t_{i,j}/n_i$ follows the Gamma distribution $\Gamma(n_i,n_i)$, we have
\begin{equation*}
    \hat{r}_i(t) \sim_{\text{i.i.d}} \exp\left(\log r_i(t) / M_{n_i}\right),
\end{equation*}
where $M_{n_i} \sim \Gamma(n_i,n_i)$. 
% We can also express it as 
% \begin{equation*}
%     \hat{r}_i(t) \sim_{\text{i.i.d}} \exp\left(\log r_i(t) / M_{n_i}\right),
% \end{equation*}
% It also corresponds the \eqref{equ:MME} with $\kappa_{i1} = 0, \kappa_{i2} = 1$ and $\bar{Z}_{n_i} = 0$.

\end{proof}

The proof of Theorem~\ref{thm:order of LCL} relies on the following lemma from \cite{newton1994bootstrap}, which establishes almost sure convergence for ergodic processes under conditional expectations.
\begin{lemma}\label{lem:recycling}
Let $ Z = (Z_1, Z_2, \dots) $ be a real-valued, stationary ergodic process, and let $ Y $ be a random vector independent of $Z$. If the real-valued measurable function $ f $ satisfies $ \mathbb{E}|f(Y, Z_1)| < \infty $, then, as $ n \to \infty $,

\[
\frac{1}{n} \sum_{i=1}^{n} f(Y, Z_i) \overset{\text{a.s.}}{\longrightarrow} \mathbb{E}[f(Y, Z_1) \mid Y]  
\] 
\end{lemma}

\begin{proof}[Proof of Theorem~\ref{thm:order of LCL}:]
Assume that the sample sizes are equal, i.e., $n_i = n$ for all $i=1,\ldots,s$.
For $x \in [0,1]$, the bootstrap version of $U(x)$ is given by
\begin{equation*}   
 U_j^*(x) = \text{P}(\hat{R}^{**}_{j,l} \leq x \mid \mathcal{D}, \mathcal{D}_{j}^*),
\end{equation*}
and the distribution function of $ U_j^*(\hat{R}) $, conditional on $\mathcal{D}$, is 
\begin{equation*}
G(x) = \text{P}(U_j^*(\hat{R}) \leq x \mid \mathcal{D}),
\end{equation*}
which is independent of $j$. 
In our Algorithm~\ref{al:dbp}, the $U_j^*(\hat{R})$ is approximated via Monte Carlo simulation, yielding the 
$\hat{u}_j$'s and the empirical approximation for $G(x)$   
\begin{equation*}
    \tilde{G}(x) = B^{-1} \sum_{j=1}^B I(\hat{u}_j \leq x).
\end{equation*}
Now we prove that, given $\mathcal{D}$ and $x$, $\tilde{G}(x) \overset{\text{a.s.}}{\longrightarrow} G(x)$ as $B, C \longrightarrow \infty$. 
Decomposing $\tilde{G}(x)$ into
\begin{equation*}
    \tilde{G}(x) = \text{P} \left(\hat{u}_j \leq x \mid \mathcal{D}, \bar{Z}_{n_i}^{**i}, M_{n_i}^{**l}, i=1,\ldots,s, l=1,\ldots,C \right) + \tilde{e}_B,
\end{equation*}
where the first term is the conditional distribution function of $\hat{u}_j$ under the randomness of the first layer auxiliary statistic, and the second term $\tilde{e}_B$ is the remaining error.

% T1
For the first term, given $j=1,\ldots,B$, we have
\begin{equation*}
     \hat{u}_j = C^{-1} \sum_{l=1}^C I\left(\hat{R}_{j,l}^{**} \leq \hat{R} \right), 
\end{equation*}
where $\hat{R}_{j,l}^{**} = \psi(\hat{r}_{1j,l}^{**}, \ldots, \hat{r}_{sj,l}^{**})$ and 
\begin{equation*}
     \hat{r}_{1j,l}^{**} = 1 - F_1 \left[ \left\{F_1^{-1}(1 - \hat{r}_{1j}^*(t)) -  \bar{Z}_{n_1}^{**l}\right\} (M_{n_1}^{**l})^{-1} \kappa_{12} + \kappa_{11} \right],
\end{equation*}
\begin{equation*}
         \hat{r}_{1j}^{*} = 1 - F_1 \left[ \left\{F_1^{-1}(1 - \hat{r}_{1}(t)) -  \bar{Z}_{n_1}^{*j}\right\} (M_{n_1}^{*j})^{-1} \kappa_{12} + \kappa_{11} \right].
\end{equation*}
The $ (\bar{Z}_{n_1}^{**l}, M_{n_1}^{**l}) $'s  are independent of $ (\bar{Z}_{n_1}^{*j}, M_{n_1}^{*j}) $.
Furthermore, each term $ I(\hat{r}_{1j,l}^{**} \leq \hat{r}_1) $  is bounded and thus $\mathbb{E} (I(\hat{r}_{1j,l}^{**} \leq \hat{r}_1)) \leq 1$. From lemma~\ref{lem:recycling}, we have 
\begin{equation*}
   \text{P} \left( \lim_{C \rightarrow \infty} \hat{u}_j = \mathbb{E}(I(\hat{R}_{j,l}^{**} \leq \hat{R})) \mid \mathcal{D}, \bar{Z}_{n_i}^{*j}, M_{n_i}^{*j}, i=1,\ldots,s   \right) = 1,
\end{equation*}
% As shown in Theorem~\ref{thm:mme}, the distribution of $ \hat{r}_1^* $ based on $ \bar{Z}_{n_1}^{*j}, M_{n_1}^{*j} $ is equivalent to the distribution of the moment-based estimate from the resamples $ \mathcal{X}_{1j}^* $, we have  
which is equivalent to
\begin{equation*}
   \text{P} \left( \lim_{C \rightarrow \infty} \hat{u}_j = U_j^*(\hat{R}) \mid \mathcal{D}, \mathcal{D}_j^*   \right)
   % = \text{P} \left( \lim_{C \rightarrow \infty} \hat{u}_j = \mathbb{E}(I(\hat{R}_{j,l}^{**} \leq \hat{R})) \mid \mathcal{D}, \bar{Z}_{n_i}^{*j}, M_{n_i}^{*j}, i=1,\ldots,s   \right) 
   = 1
\end{equation*}
This shows that $\hat{u}_j$ almost converges to $U_j^*(\hat{R})$ given $\mathcal{D}, \mathcal{D}_j^*$ for any $j=1,\ldots,B$.
By the dominated convergence theorem for conditional expectations, as $C \rightarrow \infty$, we have
\begin{equation}
\label{equ:uj}
\begin{aligned}
&\text{P} \left(\hat{u}_j \leq x \mid \mathcal{D},  \bar{Z}_{n_i}^{**l}, M_{n_i}^{**l}, i=1,\ldots,s, l=1,\ldots,C \right) \\
& \overset{\text{a.s.}}{\longrightarrow} 
\text{P} \left(U_j^*(\hat{R}) \leq x \mid \mathcal{D}, \bar{Z}_{n_i}^{**l}, M_{n_i}^{**l}, i=1,\ldots,s \right) 
= G(x).
\end{aligned}
\end{equation}
The final equivalence holds since $U_j^*$ is independent of  $ \bar{Z}_{n_i}^{**l}, M_{n_i}^{**l} $'s. 

Next, consider the remaining error term
\begin{equation}
\label{equ:eB}
\tilde{e}_B = \tilde{G}(x) - \text{P} \left(\hat{u}_j \leq x \mid \mathcal{D}, \bar{Z}_{n_i}^{**i}, M_{n_i}^{**l}, i=1,\ldots,s, l=1,\ldots,C \right).
\end{equation}
Rewrite it as 
\begin{equation*}
    \tilde{e}_B = B^{-1} \sum_{j=1}^B \tilde{e}_{B,j},
\end{equation*}
where
\begin{equation*}
\tilde{e}_{B,j} =    I(\hat{u}_j \leq x ) - \text{P} \left(\hat{u}_j \leq x \mid \mathcal{D}, \bar{Z}_{n_i}^{**l}, M_{n_i}^{**l}, i=1,\ldots,s, l=1,\ldots,C \right).
\end{equation*}
Given $\mathcal{D}$ and $\bar{Z}_{n_i}^{**l}, M_{n_i}^{**l},$ $i=1,\ldots,s,$ $l=1,\ldots,C$, the $\tilde{e}_{B,j}$'s are conditionally independent since the pairs $(\bar{Z}_{n_i}^{*j},M_{n_i}^{*j})$ are conditionally independent. Moreover, each $\tilde{e}_{B,j}$ has zero mean and satisfies $|\tilde{e}_{B,j}| \leq 1$. Therefore, we have
\begin{equation*}
\begin{aligned}
    \text{P} ( |\tilde{e}_B| \geq \epsilon) 
    &= \mathbb{E}\left\{ \text{P} \left(|\tilde{e}_B| \geq \epsilon
    \mid \mathcal{D}, \bar{Z}_{n_i}^{**l}, M_{n_i}^{**l}, i=1,\ldots, s, l=1,\ldots,C \right) \right\} \\
    &\leq (B\epsilon)^{-3} \mathbb{E}\left\{ \mathbb{E} \left\{ \left( \sum_{j=1}^B  \tilde{e}_{B,j} \right)^3
    \middle| \mathcal{D}, \bar{Z}_{n_i}^{**l}, M_{n_i}^{**l}, i=1,\ldots, s, l=1,\ldots,C \right\} \right\}\\
    &\leq (B\epsilon)^{-3} \mathbb{E} \left\{ \sum_{j=1}^B  \mathbb{E}\left(\tilde{e}_{B,j}^3\right) \right\} \leq B^{-2} \epsilon^{-3},
\end{aligned}
\end{equation*}
which applies
$$\sum_B \text{P}(|\tilde{e}_B| \geq \epsilon) < \infty.$$
By the Borel-Cantelli lemmas, we have $\tilde{e}_B \overset{\text{a.s.}}{\longrightarrow} 0 $ as $B\rightarrow \infty$. 
Combining this result with the almost sure convergence in \eqref{equ:uj}, we conclude
$\tilde{G}(x) \overset{\text{a.s.}}{\longrightarrow} G(x)$. 
Thus we have $\hat{u}_{(k)} \overset{\text{a.s.}}{\longrightarrow} G^{-1}(\alpha)$ as $B,C \rightarrow \infty$.

Finally, we have
\begin{equation*}
    \text{P}(R \geq R_{DBPT}) = \text{P} \left(U(R) \geq \hat{u}_{(k)} \right) \overset{\text{a.s.}}{\longrightarrow}
    \text{P} \left(U(R) \geq G^{-1}(\alpha) \right) = C(G^{-1}(\alpha)).
\end{equation*}
From the \cite{beran1987prepivoting}, we have $C(x) =  1 - G(x) + O(n^{-1})$.
This proves that
\begin{equation*}
    \text{P}(R \geq R_{DBPT}) = 1 - G(G^{-1}(\alpha)) + O(n^{-1}) = 1 -  \alpha + O(n^{-1}) .
\end{equation*}
This proof can be easily extended into the case of other lifetime models.
% and unequal sample size.

\end{proof}

\end{document}